\PassOptionsToPackage{dvipsnames,table}{xcolor}
\documentclass[11pt, lettersize]{article}
\usepackage{geometry}[margin = 1in]
\usepackage[utf8]{inputenc}
\usepackage[T1]{fontenc}
\usepackage{amsmath}
\usepackage{amsfonts}
\usepackage{amsthm}
\usepackage{graphicx}
\usepackage{wrapfig}
\usepackage{makecell}
\usepackage{xspace}
\usepackage{caption}
\usepackage{imakeidx}
\usepackage{thmtools}
\usepackage[shortlabels]{enumitem}

\usepackage[text size=footnotesize]{todonotes}

\usepackage[dvipsnames,table]{xcolor}

\graphicspath{{./figures/}}

\newcommand{\des}{f^{\emph{des}}}

\newcommand{\tour}[1]{\ensuremath{\mathring{#1}}}
\newcommand{\first}[1]{\ensuremath{f^{\tiny \emph{first} }(#1)}}
\newcommand{\last}[1]{\ensuremath{f^{\tiny \emph{last} }(#1)}}

\newcommand{\proj}{\hat{f}^{\emph{des}}}
\newcommand{\T}{\ensuremath{T_G}}
\newcommand{\D}{\ensuremath{T_G^\Delta}}
\newcommand{\TT}{\ensuremath{  \scalebox{1.5}{$\tau_{\scalebox{0.5}{$G$}}$}}}

\newcommand{\interval}{border\xspace}
\newcommand{\intervals}{borders\xspace}

\newtheorem{property}{Property}
\newtheorem{corollary}{Corollary}
\newtheorem{lemma}{Lemma}

\newtheorem{definition}{Definition}
\newtheorem{observation}{Observation}

\usepackage{hyperref}

\title{Worst-case Deterministic Fully-Dynamic Planar $2$-vertex Connectivity}
\author{Jacob Holm\thanks{Basic Algorithms Research Copenhagen.
University of Copenhagen, Denmark. \textit{jaho@di.ku.dk}.}, Ivor van der Hoog\thanks{
    Algorithms, Logic and Graphs. Technical University of Denmark, Denmark. \texttt{idjva@dtu.dk}.
  }, Eva Rotenberg\thanks{Algorithms, Logic and Graphs. Technical University of Denmark, Denmark. \textit{erot@dtu.dk}.}}

\begin{document}
\date{}
\maketitle
\thispagestyle{empty}

\begin{abstract}
We study dynamic planar graphs with $n$ vertices, subject to edge deletion, edge contraction, edge insertion across a face, and the splitting of a vertex in specified corners. We dynamically maintain a combinatorial embedding of such a planar graph, subject to connectivity and $2$-vertex-connectivity (biconnectivity) queries between pairs of vertices. Whenever a query pair is connected and not biconnected, we find the first and last cutvertex separating them.

Additionally, we allow local changes to the embedding by flipping the embedding of a subgraph that is connected by at most two vertices to the rest of the graph.

We support all queries and updates in deterministic, worst-case, $O(\log^2 n)$ time,
using an $O(n)$-sized data structure.

Previously, the best bound for fully-dynamic biconnectivity (subject to our set of operations) was an \emph{amortised} $\tilde{O}(\log^3 n)$ for general graphs, and algorithms with worst-case polylogarithmic update times were known only in the partially dynamic (insertion-only or deletion-only) setting.
\end{abstract}



\thispagestyle{empty}
\newpage

\section{Introduction}
In dynamic graph algorithms, the task is to efficiently update information about a graph that undergoes updates from a specified family of potential updates. Simultaneously, we want to efficiently support questions about properties of the graph or relations between vertices.
Two vertices $u$ and $v$ are $2$-vertex connected (i.e. biconnected) in a graph $G$, whenever after the removal of any vertex in $G$ (apart from $u$ and $v$) they are still connected in $G$.
This work considers dynamically maintaining a combinatorial embedding of a graph that is planar, subject to biconnectivity queries between vertices. We show how to efficiently maintain $G$ in $O(\log^2 n)$ time per update operation using linear space. 
We additionally support biconnectivity queries in $O(\log^2 n)$ time. The competitive parameters for dynamic algorithms include update time, query time, the class of allowed updates, the adversarial model, and whether times are worst-case or amortized. 
 We present a deterministic algorithm: which means that all statements hold in the strictest adversarial model; against adaptive adversaries. Interestingly, for general graphs, there seems to be a large class of problems for which the deterministic amortized algorithms grossly outperform the deterministic worst-case time algorithms:
for dynamic connectivity 
the state-of-the-art worst-case update time is of the form 
$O(n^{o(1)})$~\cite{DBLP:conf/soda/GoranciRST21},
whilst the state-of-the-art amortized update time is $\tilde{O}(\log^2 n)$~\cite{Holm:2001,Wulff-Nilsen16};
for planarity testing, the best amortized solution has $O(\log^3 n)$~\cite{DBLP:conf/stacs/HolmR21} update time, compared to 
$O(n^{2/3})$ worst-case~\cite{GalilIS99} (in a restricted setting).
For biconnectivity in general graphs the current best worst-case solution has update time $O(\sqrt{n})$ \cite{Eppstein97}. The best amortized update time is $\tilde{O}(\log^3 n)$ \cite{Holm:2001,Thorup:2000,Holm18a}. For plane graphs, Henzinger~\cite{henzinger2000improved} shows how to support biconnenctivity queries for a plane graph in $O(\log n)$ query time and $O(\log^2 n)$ update time (where the updates may be edge deletions and insertions across a face in the plane embedding). This algorithm is deterministic and worst-case.

In this work, we provide algorithms for updating connectivity information of a combinatorially embedded planar graph, that is both deterministic, worst-case, and fully-dynamic. 

\begin{restatable}{theorem}{theoremmain}
\label{thm:main_full}
We maintain a planar combinatorial embedding in $O(\log^2 n)$ time subject to:
\begin{itemize}[noitemsep, nolistsep]
    \item $\operatorname{delete}(e)$: where $e$ is an edge, deleting the edge $e$,
    \item $\operatorname{insert}(u,v,f)$: where $u,v$ are incident to the face $f$, inserting an edge $uv$ across $f$,
    \item $\operatorname{find-face}(u,v)$: returns some face $f$ incident to both $u$ and $v$, if any such face exists.
    \item $\operatorname{contract}(e)$: where $e$ is an edge, contract the edge $e$,
    \item $\operatorname{split}(v,c_1,c_2)$: where $c_1$ and $c_2$ are corners (corresponding to gaps between consecutive edges) around the vertex $v$, split $v$ into two vertices $v_{12}$ and $v_{21}$ such that the edges of $v_{12}$ are the edges of $v$ after $c_1$ and before $c_2$, and $v_{21}$ are the remaining edges of $v$,
    \item $\operatorname{flip}(v)$: for a vertex $v$: flip the orientation of the 
    connected component containing $v$.
\end{itemize}
\noindent
We may answer the following queries in $O(\log ^2 n)$ time:
\begin{itemize}[noitemsep, nolistsep]
    \item $\operatorname{connected}(u,v)$, where $u$ and $v$ are vertices, answer whether they are connected,
    \item $\operatorname{biconnected}(u,v)$, where $u$ and $v$ are connected, answer whether they are biconnected.
    When not biconnected, we may report the separating cutvertex closest to $u$.
\end{itemize}
\end{restatable}

\noindent
Our update time of $O(\log ^2 n)$ should be seen in the light of the fact that even just supporting edge-deletion, insertion, and $\operatorname{find-face}(u,v)$, currently requires $O(\log ^2 n)$ time~\cite{ItalianoPR93}.
We briefly review the concepts in this paper and the state-of-the-art.

\textbf{Biconnectivity.} For each connected component of a graph, the cutvertices are vertices whose removal disconnects the component. These cutvertices partition the edges of the graph into \emph{blocks} where each block is either a single edge (a \emph{bridge} or \emph{cut-edge}), or a biconnected component. A pair of vertices are biconnected if they are incident to the same biconnected component, or, equivalently, if there are two vertex-disjoint paths connecting them~\cite{menger1927allgemeinen}. 
This notion generalises to $k$-connectivity where $k$ objects of the graph are removed.
While $k$-edge-connectivity is always an equivalence relation on the vertices, $k$-vertex-connectivity happens to be an equivalence relation for the edges only when $k\leq 2$.


\textbf{Dynamic higher connectivity} aims to facilitate queries to $k$-vertex-connectivity or $k$-edge-connectivity as the graph undergoes updates. For two-edge connectivity and biconnectivity in general graphs, there has been a string of work~\cite{Frederickson97,Henzinger95,Eppstein97,Henzinger97,Holm:2001,Thorup:2000,Holm18a}, and the current best deterministic results have $O(\log^2 n\log\log ^2 n)$ amortized update time for $2$-edge connectivity~\cite{Holm18a}, and spend an additional amortized $\log(n)$-factor for biconnectivity~\cite{Holm:2001,Thorup:2000}.  
Thus, the current state of the art for deterministic two-edge connectivity is $\log\log (n)$-factors away from the best deterministic connectivity algorithm~\cite{Wulff-Nilsen16a}, while deterministic biconnectivity is $\log (n)$-factors away. See \cite{Frederickson85, HeTh97, Henzinger:1999, Holm:2001, Thorup:2000, Kapron:2013, HuangHKP17, Wulff-Nilsen16a, kejlbergrasmussen_et_al:LIPIcs:2016:6395, NSW17, DBLP:conf/soda/GoranciRST21} for more work on dynamic connectivity. 
For $k$-(edge-)connectivity with $k > 2$, only partial results have appeared, including incremental~\cite{PoutreW98,BattistaT96, PoutreLO93,Poutre00} and decremental~\cite{GiammarresiI96,Thorup99,HolmIKLR18,Aamand21} results.
The strongest lower bound is by P{\u{a}}tra{\c{s}}cu et al.~\cite{patrascu06}, and implies that of update- and query time cannot both be $o(\log n)$ for any of the mentioned fully dynamic problems on general graphs, and this holds even for planar embedded graphs.
For special graph classes, such as planar graphs, graphs of bounded genus, and minor-free graphs, there has been a bulk of work on connectivity and higher connectivity, e.g.~\cite{EppItaTam92,Hershberger94,Giammarresi:96,Gustedt98,Eppstein99,Lacki2011,Lacki15,Holm17,Holm18b}. 
For dynamic planar embedded graphs, the deterministic poly-logarithmic worst-case algorithm for two-edge connectivity dates back three decades~\cite{Hershberger94}. In this paper, we obtain the same $O(\log ^2 n)$ bound as in \cite{Hershberger94}, 
for the harder problem of biconnectivity. An open question remains whether higher connectivity can generally be maintained in polylogarithmic worst-case time for dynamic planar graphs ($k$-connectivity and $k$-edge connectivity, $k>2$). 

\textbf{Techniques.} Exploiting properties of planar graphs, we use the tree-cotree decomposition: a partitioning of edges into a spanning tree of the graph and a spanning tree of its dual. Using tree-cotree compositions to obtain fast dynamic algorithms is a technique introduced by Eppstein~\cite{Eppstein:2003}, who obtains algorithms for dynamic graphs that have efficient genus-dependent running times. Note that the construction in~\cite{Eppstein:2003} does not facilitate inserting edges in a way that minimises the resulting genus. Such queries are, however, allowed in the structure by Holm and Rotenberg~\cite{Holm17}, which also utilises the tree-cotree decomposition. 

On this spanning tree and cotree, we use top-trees to handle local biconnectivity information. Much of our work concerns carefully choosing which biconnectivity information is relevant and sufficient to maintain, as top-tree clusters are merged and split. Note that the ideas for two-edge connectivity introduced by Hershberger et al. in~\cite{Hershberger94}, i.e. ideas of using topology trees on a vertex-split version of the graph to keep track of edge bundles, do not transfer to the problem at hand, since vertex-splitting changes the biconnectivity structure. 

 
%
\section{Preliminaries}

We study a dynamic plane embedded graph $G = (V, E)$, where $V$ has $n$ vertices. We assume access to $G$ and some \emph{combinatorial embedding}~\cite{eppstein2003dynamic} of $G$ that specifies for every vertex in $G$ the cyclical ordering of the edges incident to that vertex.
Throughout the paper, we maintain some associated spanning tree $\T = (V, E')$ over $G$. 
We study the combinatorial embedding subject to the update operations specified in Theorem~\ref{thm:main_full}. This is the same setting and includes the same updates as by Holm and Rotenberg~\cite{holm2017dynamic}. 

\textbf{Spanning and co- trees.}
If $G$ is a connected graph, a spanning tree $\T$ is a tree where its vertices are $V$, and the edges of $\T$ are a subset of $E$ such that $\T$ is connected.
Given $\T$, the cotree  $\D$ has as vertices the faces in $G$, and as edges of $\D$ are all edges in $G \setminus \T$.
It is known that the cotree is a spanning tree of the dual graph of $G$~\cite{eppstein1992maintenance}.

\textbf{Induced graph.} We adopt the standard notion of (vertex) induced subgraphs: for any $V' \subseteq V$, $G[V']$ is the subgraph created by all edges  $e \in E$ with both endpoints of $e$ in $V'$. 
For any $G$ and $V'$, we denote by $G \setminus G[V']$ the graph $G$ minus all edges in $G[V']$.
Observe that for $(V_1, V_2)$ the set $G[V_1 \cup V_2]$ is not necessarily equal to $G[V_1] \cup G[V_2]$ (Figure~\ref{fig:edgeinequality}). 

\textbf{Top trees.}
Our data structure maintains a specific variant of a \emph{top tree} $\TT$ over the graph $G$~\cite{alstrup2005maintaining, tarjan2005self, holm2017dynamic}. This data structure is a hierarchical decomposition of a planar, embedded, graph $G$ based on a spanning tree $\T$ of $G$.
Formally, for every connected subgraph $S$ of $\T$ we define the \emph{boundary vertices} of $S$ as the vertices incident to an edge in $\T\setminus S$.  A \emph{cluster} is a connected subgraph of $\T$ with at most $2$ boundary vertices.
A cluster with one boundary vertex is a \emph{point cluster}; otherwise a \emph{path cluster}.
A top tree $\TT$ is a hierarchical decomposition of $G$ (with depth $O(\log n)$) into point and path clusters that is structured as follows: 
the leaves of $\TT$ are the path and point clusters for each edge $(u, v)$ in $\T$ (a leaf in $\TT$ is a point cluster if and only if the corresponding edge $(u, v)$ is a leaf in $\T$).
Each inner node $\nu \in \TT$ \emph{merges} a constant number of \emph{child} clusters sharing a single vertex into a new point or path cluster. The vertex set of $\nu$ is the union of those corresponding to its children. 
We refer to combining a constant number of nodes into a new inner node as a \emph{merge}. We refer to its inverse as a \emph{split}.
Furthermore, for planar embedded graphs, we restrict our attention to \emph{embedding-respecting} top trees; that is, given for each vertex a circular ordering of its incident edges, top trees that only allow merges of neighbouring clusters according to this ordering. In other words, if two clusters $\nu$ and $\mu$ share a boundary vertex $b$, and are mergable according to the usual rules of top trees, we only allow them to merge if furthermore they contain a pair of neighbouring edges $e_{\mu}\in\mu$ and $e_{\nu}\in\nu$ where $e_{\mu}$ is a neighbour of $e_{\nu}$ around $b$. 
Holm and Rotenberg~\cite{holm2017dynamic} show how to dynamically maintain $\TT$ (and the spanning tree and cotree) with the following property:

\begin{property}
\label{prop:ordering}
Let $\nu \in \TT$ be a point cluster with boundary vertex $u$. 
The graph $G[\nu]$ is a contiguous segment of the extended Euler tour of $\T$.
\end{property}

\begin{corollary}
\label{corollary:ordering}
Let $\nu \in \TT$ be a point cluster with boundary vertex $u$. The edges of $G[\nu]$ that are incident to $u$ form a connected interval in the clockwise order around $u$.
\end{corollary}

\begin{figure}[b]
  \centering
  \includegraphics{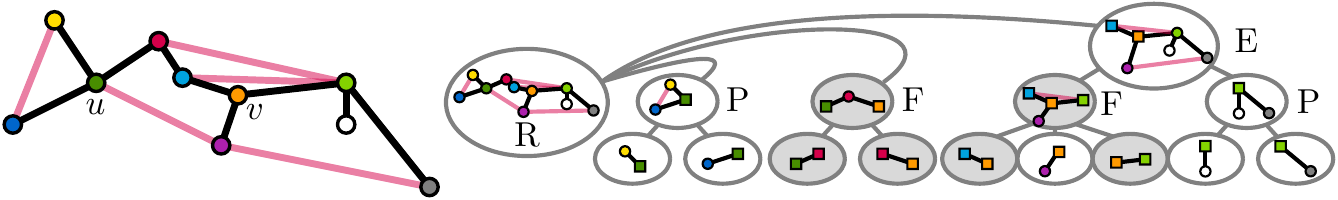}
  \caption{
    We recursively decompose $G$ based on a spanning tree $\T$. 
    Square vertices are boundary vertices. We highlight path clusters.
    The root node has three children, where one is a path cluster that exposes $\{ u, v \}$. The letters indicate the later defined merge type.
  }
  \label{fig:toptree}
\end{figure}

\textbf{Edge division.} For ease of exposition, we perform the trick of subdividing edges into paths of length three. We refer to $G^o$ as the original and $G$ as this edge-divided graph. Since $G^o$ is planar, this does not asymptotically increase the number of vertices.  We note:
\begin{enumerate}[noitemsep, nolistsep]
    \item Edge subdivision respects biconnectivity (since edge subdivision preserves the cycles in the graph; it preserves biconnectivity).
    \item Any spanning tree of $G^o$ can be transformed into a spanning tree of $G$ where all non-tree edges have end points of degree two: for each non-tree edge in $G$, include exactly the first and last edge on its corresponding path in the spanning tree. This property can easily be maintained by any dynamic tree algorithm.
    \item Dynamic operations in $G^o$ easily transform to a constantly many operations in $G$.
\end{enumerate}
With this in place, our top tree structure automatically maintains more information about the endpoints of non-tree edges and their ordering around each endpoint. 

\textbf{Paper notation.}
We refer to vertices in $G$ with Latin letters. We refer to nodes in the top tree $\TT$ with Greek letters. We refer indistinguishably to nodes $\nu \in \TT$ and their associated vertex set. 
Vertices $u$ and $v$ are boundary vertices. 
For a path cluster $\nu \in \TT$ with boundary vertices $\{ u, v \}$ we call its \emph{spine} $\pi(\nu)$ the path in $\T$ that connects $u$ and $v$.
For any path, its \emph{internal vertices} exclude the two endpoints. 
For a point cluster with boundary vertex $u$, its spine $\pi(\nu)$ is $u$. 
We denote by $\TT(\nu)$ the subtree rooted at $\nu$.

\textbf{Slim-path top trees over $G$.}
We use a variant of the a top tree called a \emph{slim-path top tree} by Holm and Rotenberg~\cite{holm2017dynamic}. 
This variant of top trees upholds the \emph{slim-path invariant}: for any path-cluster $\nu$, all edges (of the spanning tree $\T$) in the cluster that are incident to a boundary vertex
belong to the spine. 
In other words: for every path cluster $\nu \in \TT$, for each boundary vertex $u$, there is exactly one edge in the induced subgraph $G[\nu]$ that is connected to $u$. 
The root of this top tree is the merge between a path cluster with boundary vertices $u$ and $v$, with at most two point clusters $\lambda$, $\mu$, with $\pi(\lambda) = \{ u \}$ and $\pi(\mu) = \{v \}$.\footnote{In the degenerate case where the graph is a star, we add one dummy edge to $G$ to create a path cluster.} 
Holm and Rotenberg show how to obtain (and dynamically maintain) this top tree with four types of merges between clusters, illustrated by Figure~\ref{fig:mergetypes} and~\ref{fig:invariant_one_large}. Our operations merge:
\begin{description}[nosep,labelindent=\parindent]
\item[(Root merge)] at most two point clusters and a path cluster to create the root node,

\item[(Point merge)] two point clusters $\mu, \nu$ with $\pi(\mu) = \pi(\nu)$.
    
\item[(End merge)] a point and a path cluster that results in a point cluster, and
    
\item[(Four-way merge)] two path clusters $\mu, \nu$ and at most two point clusters $\alpha, \beta$,     where their common intersection is one \emph{central} vertex $m$. If there are two point clusters, they are not adjacent around $m$. This merge creates a path cluster.
\end{description}

\noindent    
Holm and Rotenberg~\cite{holm2017dynamic} dynamically maintain the above data structure with at most $O(\log n)$ merges and splits per graph operation (where each merge or split requires $O(\log n)$ additional operations). Their data structure supports two additional crticical operations:

\begin{description}[nosep, labelindent=\parindent]
\item[Expose$(u, v)$] selects two vertices $u,v$ of $G$ and ensures that for the unique path cluster 

$\nu$ of 
the root node, $u$ and $v$ are the two endpoints of $\pi(\nu)$; ($O(\log n)$ splits/merges).

\item[Meet$(u, v, w)$] selects three vertices $u, v, w$ of $G$ and returns their meet in $\T$, defined as the unique common vertex on all $3$ paths between the vertices. Moreover, they also support this operation on the cotree $\D$; ($O(\log n)$ time).

\end{description}

\begin{figure}[b]
  \centering
  \includegraphics{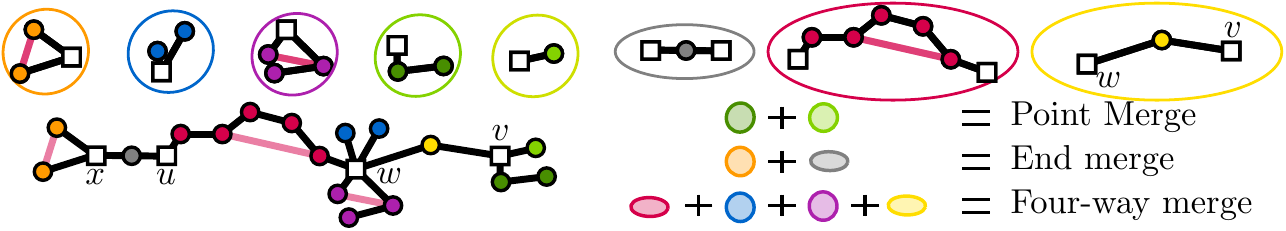}
  \caption{
    A graph $G$ which we already split into five point clusters (circles) and three path clusters (ovals). 
    We show the combinations of clusters that create three merge types. To obtain the root: execute the three suggested merges and merge the remaining components.
  }
  \label{fig:mergetypes}
\end{figure}

\noindent
When $\nu$ is a node in an embedding-respecting slim-path toptree, and $G$ is formed from $G^o$ via edge-subdivisions, note that $G[\nu]$ has the following properties: \label{inducedII}
\begin{itemize}[noitemsep, nolistsep]
    \item For a point cluster $\nu$ with boundary vertex $b$ that encompasses the tree-edges $e_1$ to $e_l$ in the circular ordering around $b$, $G[\nu]$ corresponds to the sub-graph in $G^o$ induced by all vertices in $\nu$, except edges to $b$ that are not in $e_1,\ldots,e_l$.
    \item When $\nu$ is a path cluster, $G[\nu]$ corresponds to the sub-graph in $G^o$ induced by all vertices in $\nu$ except non-tree edges incident to either of the two boundary vertices.
\end{itemize}


\section{Dynamic biconnectivity queries and structure }
\label{sec:datastructure_proofs}

\begin{figure}[b]
  \centering
  \includegraphics{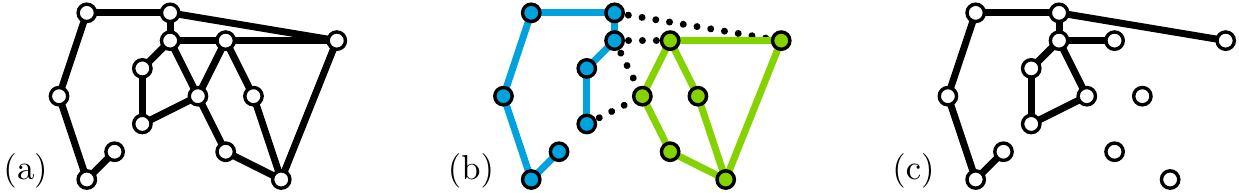}
  \caption{
    (a) A planar graph $G$. (b) A set $\beta \subset V$ and $\gamma \subset V$. We show $G[\beta]$ and $G[\gamma]$. The set $G[\beta] \cup G[\gamma]$ contains no black edges.  (c) We show $G \setminus G[\gamma]$.
  }
  \label{fig:edgeinequality}
\end{figure}

We want  maintain a slim-path top tree, subject to the aforementioned operations, that additionally supports biconnectivity queries in $O(\log^2 n)$ time.    Our data structure consists of the slim-path top tree from~\cite{holm2017dynamic} with three invariants which we define later in this section:
\begin{enumerate}[noitemsep, nolistsep]
    \item For each cluster $\nu \in \TT$, we store the biconnected components in $G[\nu]$ which are relevant for the exposed vertices $(u, v)$. 
    \item For each cluster $\nu$, we store the information required to navigate through the top tree.
    \item For each stored biconnected component, we store its `border' along the spine $\pi(\nu)$.
\end{enumerate}
Here, we show the technical details required to define these invariants.
We specify for each $\nu \in \TT$, for each endpoint $u$ of the spine, a \emph{designated face}:
\begin{restatable}{lemma}{uniquePathFace}
\label{lemma:unique_path_face}
Let $\nu \in \TT$ be a path cluster. Each boundary vertex $u$ (resp. $v$) is incident to a unique face $\des_u(\nu)$ (resp. $\des_v(\nu)$) of $G[\nu]$.
Moreover, all edges in $G$ that are not in $G[\nu]$ are contained in either $\des_u(\nu)$ or $\des_v(\nu)$.\footnote{Formally, we can say that an edge, vertex, or face in $G$ is contained in face $f$ of a subgraph $G'$, if it is contained in $f$ in any drawing of $G$ that is consistent with the current combinatorial embedding.}
\end{restatable}
\begin{proof}
By definition of our subgraph induced by a path cluster, each boundary vertex $u$ is incident to at most one edge in $G[\nu]$ and therefore incident to a unique face in $G[\nu]$. The graph $T \setminus G[\nu]$ consists of at most two subtrees, rooted at $u$ and $v$. 
Since $u$ and $v$ are incident to exactly one face each, these subtrees must be contained within $\des_u(\nu)$ and $\des_v(\nu)$ respectively.
\end{proof}

\begin{restatable}{lemma}{uniquePointFace}
\label{lemma:unique_point_face}
Let $\nu \in \TT$ be a point cluster with boundary vertex $u$.
The subgraph $G[\nu]$ has a unique face $\des_u(\nu)$ such that all edges in $G$ that are not in $G[\nu]$ are contained in $\des_u(\nu)$.
\end{restatable}

\begin{proof}
This follows immediately from Corollary~\ref{corollary:ordering}.
\end{proof}

\begin{restatable}{corollary}{containment}
\label{cor:contaiment}
Let $\nu \in \TT$ have a boundary vertex $u$, let $\mu$ be a descendent of $\nu$ and $x$ be the boundary vertex of $\mu$ closest to $u$ in $T$. Then $\des_u(\nu) \subseteq \des_x(\mu)$.
\end{restatable}

\noindent
Lemmas~\ref{lemma:unique_path_face} and~\ref{lemma:unique_point_face} inspire the following definition:
for all $\nu \in \TT$, for each boundary vertex $u$ (or $v$) of $\nu$, there exists a unique face $f$ which we call its \emph{designated face} $\des_u(\nu)$ (or $\des_v(\nu)$). 
Intuitively, we are only interested in biconnected components that are edge-incident to $\des_u(\nu)$ or $\des_v(\nu)$.
Let for a node $\nu$ with boundary vertex $u$ a biconnected component $B$ be edge-incident to $\des_u(\nu)$. Let $\mu$ be a descendant of $\nu$ and $B \subseteq G[\mu]$ then, by Corollary~\ref{cor:contaiment}, $B$ must be edge-incident to $\des_x(\mu)$ where $x$ is the boundary vertex of $\mu$ closest to $u$. Formally, in this scenario, we define the \emph{projected face} of $u$ in $\mu$ as $\proj_u(\mu) = \des_x(\mu)$.

\textbf{Relevant and alive biconnected components.}
Consider for a cluster $\nu$, a biconnected component $B$ of the induced subgraph $G[\nu]$.  We say that $B$ is \emph{relevant} with respect to $\nu$ if $B$ is vertex-incident to the spine $\pi(\nu)$.
We say that $B$ is \emph{alive} with respect to a face $f$ in $G[\nu]$ if $B$ is edge-incident to $f$. 
We denote by $BC(\nu, f)$ the set of biconnected components in the induced subgraph $G[\nu]$ that are relevant with respect to $\nu$ and alive with respect to $f$.
Intuitively, we want to keep track of the relevant and alive components (with respect $\des_u(\nu)$ or $\des_v(\nu)$). To save space, we store only the relevant biconnected components of $\nu$ that are not in its children. Formally (Figure~\ref{fig:invariant_one_large}), we define an invariant:

\begin{restatable}{invariant}{relevantBicomp}
\label{inv:relevant_bicomp}
For each cluster $\nu \in \TT$ (apart from the root) with children $\mu_1,  \mu_2, \ldots 
\mu_s$ where $u$ is a boundary vertex of $\nu$, we store a unique object for each element in: \vspace{-0.3cm} \[
BC^*_u(\nu) := BC(\nu, \des_u(\nu)) \setminus {\textstyle\bigcup_{i=1}^s} BC(\mu_i,\proj_u(\mu_i)).\]
\end{restatable}

\begin{figure}[t]
  \centering
  \includegraphics{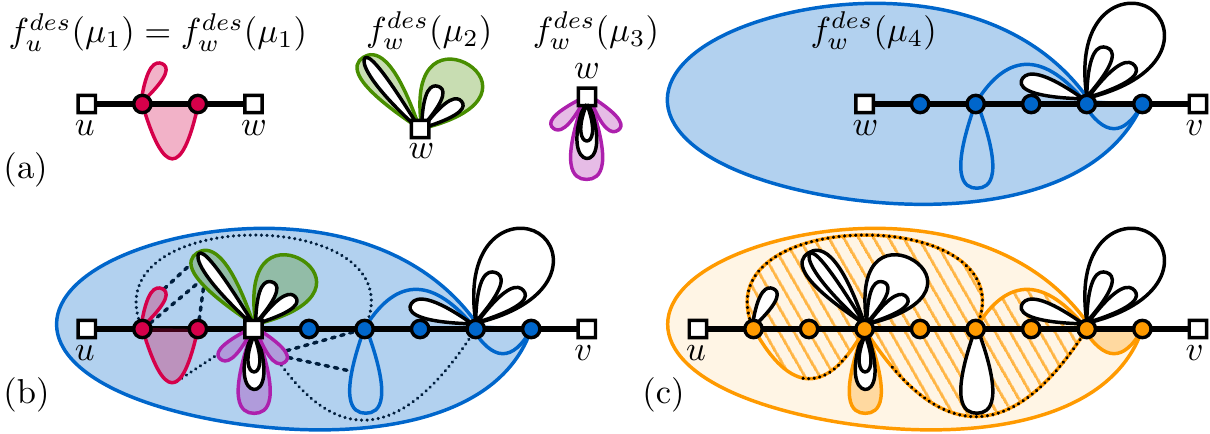}
  \caption{
    (a) Nodes $\mu_i$ with faces $\des_u(\mu_1)$ and $\des_w(\mu_i)$.
    We color $BC_u(\mu_1, \des_u(\mu_1))$ and $BC_w(\mu_i, \des_w(\mu_i))$.
    (b) The four-way merge introduces a new node $\nu$. Edges in $G[\nu]$ that are not the induced subgraph of $G[\mu_j]$ for $j \in \{1,2,3,4\}$ are dashed/dotted. 
    (c) Every dotted edge is part of a new biconnected component $B \in BC^*_u(\nu)$ which we show as tiled.
  }
  \label{fig:invariant_one_large}
\end{figure}

\noindent
Storing biconnected components in this way does not make us lose information:

\begin{restatable}{lemma}{uniquedescendent}
\label{lemma:unique_descendent}
Let $\nu \in \TT$ with boundary vertex $u$ and $B \in BC(\nu, \des_u(\nu))$. There exists a unique node $\mu$ in $\TT(\nu)$ where: 
$B \in BC^*_x(\mu)$ and $x$ is the closest boundary vertex of $\mu$ to $u$.
\end{restatable}

\begin{proof}
Each $B' \in BC(\nu, \des_u(\nu))$ is a maximal set of edges that are part of a biconnected component in $G[\nu]$. 
We denote for every descendant $\mu_i$ of $\nu$ by $x_i$ the boundary vertex of $\mu$ closest to $u$. 
Per definition, either $B \in BC^*_u(\nu)$, or $B \in BC_{x_i}(\mu_i, \des_{x_i}(\mu_i))$ for some child $\mu_i$ of $\nu$. 
Thus, there exists a descendant $\mu_j$ of $\nu$ with $B \in BC^*_{x_j}(\mu_j)$. We show it is unique: 

Per definition, for every pair of children $\mu_i, \mu_j$, the graphs $G[\mu_i]$ and $G[\mu_j]$ share no edges. Thus, there exists at most one root-to-leaf path in $\TT$ of nodes $\mu_\ell$ where $B \in G[\mu_\ell]$.

Suppose for the sake of contradiction that there are two nodes $\mu_a, \mu_b \in \TT$ where $B \in BC^*_{x_a}(\mu_a)$ and $B \in BC^*_{x_b}(\mu_b)$. Without loss if generality we assume that $\mu_a$ is an ancestor of $\mu_b$. 
However, $B \in BC^*_{x_b}(\mu_b)$ implies that for the parent $\mu_c$ of $\mu_b$, either $B \in BC(\mu_c, \des_{x_c}(\mu_c))$, or that $B$ is contained in some larger biconnected component $B' \in BC(\mu_c, \des_{x_c}(\mu_c))$. Recursive application of this argument creates a contradiction with the assumption that $B \in BC^*_{x_a}(\mu_a)$, or that  $B$ is a maximal biconnected set in $G[\nu]$.
\end{proof}

Throughout this paper we show that for each cluster $\nu \in \TT$ with boundary vertex $v$, the set $BC^*_v(\nu)$ contains constantly many elements. 
Moreover, the root vertices $(u, v)$ are biconnected at the root $\mu$ if and only if they share a biconnected components in $BC^*_u(\mu)$. 
In this section, we define two additional invariants so that we can maintain Invariant~\ref{inv:relevant_bicomp} in $O(\log n)$ additional time per split and merge. This will imply Theorem~\ref{thm:main_full}.
\subsection{The data structure}
Before we specify our data structure, we first define some additional concepts and notation.
The core of our data structure is a slim-path top tree $\TT$ on $G$, that supports the expose operation in $O(\log n)$ splits and merges and the meet operation in both the spanning tree $\T$ and its cotree $\D$ in $O(\log n)$ time. We add the following invariant for $\TT$ where each:

\begin{restatable}{invariant}{upwardpointers}
\label{inv:upwardpointers}
\begin{enumerate}[(a), ref={\ref{inv:upwardpointers}(\alph*)}, noitemsep, nolistsep]
    \item\label{inv:upwardpointers-boundary-parent} node $\nu$ has pointers to its boundary vertices and parent node.
    \item\label{inv:upwardpointers-path}  path cluster $\nu$ stores: the length of  $\pi(\nu)$ and its outermost spine edges.
    \item\label{inv:upwardpointers-point}  point cluster $\nu$ stores: the number of tree edges in $\nu$ incident to the boundary vertex.
    \item\label{inv:upwardpointers-lca} $x \in V$ points to the lowest common ancestor in $\TT$ where $x$ is a boundary vertex.  
\end{enumerate}
\end{restatable}
\noindent
Finally we add one final invariant which uses three additional concepts: slices of biconnected components, index orderings on the spine and an orientation on edges incident to a spine.

\textbf{Biconnected component slices.}
For any node $\nu \in \TT$ and any biconnected component $B\subseteq G[\nu]$, its \emph{slice} is the interval $B \cap \pi(\nu)$ (which may be empty, or one vertex):

\begin{restatable}{lemma}{endpoints}
\label{lemma:endpoints}
Let $\nu \in \TT$.
For all (maximal) biconnected components $B$ in $G[\nu]$, if $B \cap \pi(\nu)$ is not empty, it is a path (possibly consisting of a single vertex).
\end{restatable}

\begin{proof}
Consider a traversal of $\pi(\nu)$, the first and last vertex $a$ and $b$ of $\pi(\nu)$ that are in $B$ and their subpath $\pi'$.
Since $a$ and $b$ are biconnected, there exists two internally vertex-disjoint paths from $a$ to $b$. Any internal vertex $c\in \pi'$ can be on at most one of these paths, and so can not separate $a$ from $b$. If $c$ is on one of the paths it is clearly biconnected to both $a$ and $b$. Otherwise let $a'$,$b'$ be the first vertex on the path from $c$ to $a$ (resp. $b$) that is on the cycle. Deleting any one vertex leaves  $c$ connected to at least one of $a'$ and $b'$, and through that vertex $c$ is connected to both $a$ and $b$.
Thus, any $c$ on $\pi'$ must be in $B$. 
\end{proof}

\textbf{Index orderings.}
For any node $\nu \in \TT$ and spine vertex $w \in \pi(\nu)$, let $w$ be the $i$'th vertex on $\pi(\nu)$.  We say that $i$ is the \emph{index} of $w$ in $\pi(\nu)$ and show:

\begin{restatable}{lemma}{indexing}
\label{lemma:indexing}
Given Invariant~\ref{inv:upwardpointers}, a path cluster $\nu$ and a vertex $w\in\pi(\nu)$, we can compute for every path cluster $\beta$ with $\pi(\beta)\subseteq\pi(\nu)$ the index of $w$ in $\pi(\beta)$ in $O(\log n)$ total time.
\end{restatable}

\begin{proof}
By Invariant~\ref{inv:upwardpointers-boundary-parent}, we have a pointer to the boundary vertices $u$ and $v$ of $\nu$. 
We claim we can obtain an edge $e_w$ of $\pi(\nu)$ incident to $w$ in constant time.
Indeed, in the special case where $w$ is either $u$ or $v$ then by Invariant~\ref{inv:upwardpointers-path} we have a pointer to the spine edge $e_w$ incident to $w$. 
If $w$ is an internal vertex of $\pi(\nu)$, there is a unique four-way merge where $w$ is the central vertex.
By Invariant~\ref{inv:upwardpointers-lca}, we obtain a pointer to this merge in $O(1)$ time. $w$ must be a boundary vertex of a path cluster in this merge and so we obtain $e_w$.

The edge $e_w$ corresponds to a leaf vertex $\mu$ in $\TT$ and the  index of $w$ in $\mu$ is either $0$ or $1$. Since $e_w$ is part of the spine, the path from $\mu$ to $\nu$ in $\TT$ consists of only four-way merges (indeed, any point merge or end merge would result in a spine without $e_w$).
So, consider each four-way merge with path clusters $\alpha_1, \alpha_2$ into a path cluster $\beta$, on this path in $\TT$. 
If $w \in \alpha_1$ then the index of $w$ in $\pi(\beta)$ is equal to the index of $w$ in $\pi(\alpha_1)$. Otherwise, this index increases by the length of $\pi(\alpha_1)$ (which we stored according to Invariant~\ref{inv:upwardpointers}). 
Hence, since $\TT$ has height $O(\log n)$, we obtain the index of $w$ in each $\pi(\beta)$ with $O(\log n)$ additions.
\end{proof}

\noindent
Similarly in a point cluster $\nu$, for each tree edge in $\nu$  incident to the boundary vertex $u$ we  define its \emph{clockwise index} in $\nu$ as its index in the clockwise ordering of the edges around $u$. 

\begin{restatable}{lemma}{clockwise}
\label{lem:cwindex}
  Given Invariant~\ref{inv:upwardpointers}, and an edge $e=(u, x)\in\T$, we can compute the clockwise index of $e$ in every point cluster $\nu\in\TT$ that contains $e$ and has $u$ as boundary vertex, simultaneously, in worst case $O(\log n)$ total time.
\end{restatable}

\begin{proof}
  Starting from the leaf $\nu_e$ in $\TT$ that contains $e$ and traversing up the tree, all such clusters form a contiguous subsequence of the clusters visited. If any such cluster exists, there is a first ancestor $\nu_1$ to $\nu_e$ that is a point cluster, and this must have $u$ as boundary node. The clockwise index of $e$ in $\nu_1$ is $0$. For each proper ancestor $\nu_i$ of $\nu_1$ that is a point cluster with $u$ as boundary node, both children are point clusters, and one of the children ($\nu_{i-1}$) contains $e$. If that is the first child in clockwise order, the clockwise index of $e$ in $\nu_i$ is the same as in $\nu_{i-1}$. Otherwise it is the index in $\nu_{i-1}$ plus the number of tree edges incident to $u$ in the other point cluster child. In each case we compute the clockwise index of $e$ in $\nu_i$ using at most one addition. Since the height of $\TT$ is $O(\log n)$, computing all the clockwise indices of $e$ around $u$ takes $O(\log n)$ time. 
\end{proof}

\textbf{Euler tour paths and endpoint orientations.}
Consider the Euler tour of $\T$ and an embedding of that Euler tour such that the Euler tour is arbitrarily close to the edges in $\T$. 
Each edge $e$ in $G$ that is not in $\T$ must intersect the Euler tour twice. 
We classify each endpoint of $e$ based on where it intersects this Euler tour. Formally, we define (Figure~\ref{fig:edgeclassification_large}):

\begin{figure}[t]
  \centering
  \includegraphics{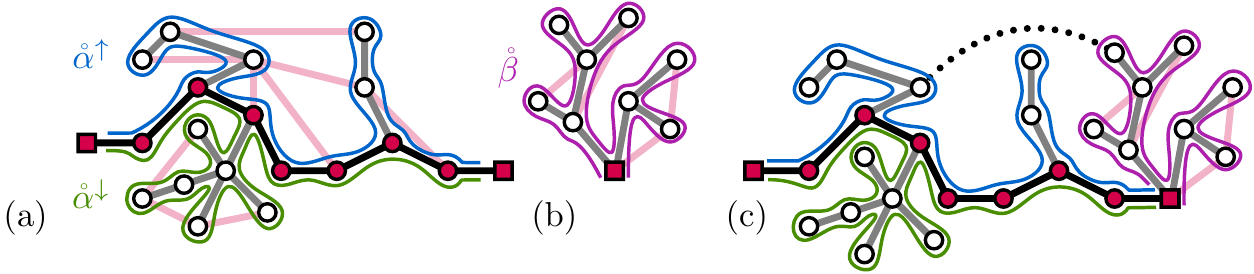}
  \caption{
    (a) A path cluster $\alpha$ with $\pi(\alpha)$ in black and edges in $\T$ in black or grey. We show the tourpaths $\tour{\alpha}^\uparrow$ and $\tour{\alpha}^\downarrow$ in blue and green. 
    (b) A point cluster $\beta$ with the path $\tour{\beta}$. 
    (c) Any edge in $G[\alpha \cup \beta] \setminus (G[\alpha]\cup G[\beta])$ must intersect 
    one of $\{ \tour{\alpha}^\uparrow, \tour{\alpha}^\downarrow \}$ and $\tour{\beta}$.
  }
  \label{fig:edgeclassification_large}
\end{figure}

\begin{definition}
For a point cluster $\nu$ with boundary vertex $u$, we denote by $\tour{\nu}$ its \textbf{tourpath} (the 
segment of the 
Euler tour in $\T$ from $u$ to $u$ that is incident to edges in $G[\nu]$).
\end{definition}

\begin{definition}
For a path cluster $\nu$ with boundary vertices $u$ and $v$. We denote by  $\tour{\nu}^\uparrow$ and $\tour{\nu}^\downarrow$ its two \textbf{tourpaths} (the two paths in the Euler tour in $\T$ from $u$ to $v$).
\end{definition}

\begin{definition}
For any tourpath $\tour{\alpha}$ let $e_1$ be the first edge of $\T$ incident to $\tour{\alpha}$ and $e_2$ be the last edge incident to $\tour{\alpha}$. 
We denote by $\first{\tour{\alpha}}$ the unique face in $G$ (incident to $e_1$) whose interior contains the start of $\tour{\alpha}$. The face $\last{\tour{\alpha}}$  is defined analogously using $e_2$.
\end{definition}

\begin{figure}[h]
  \centering
  \includegraphics{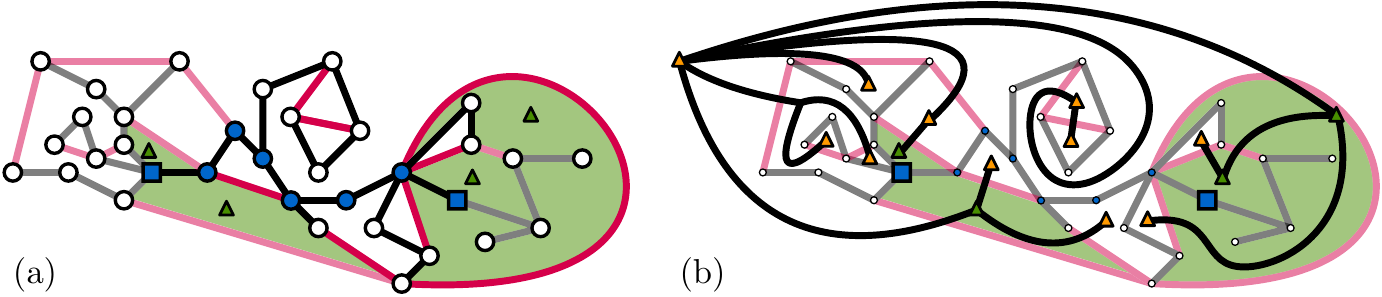}
  \caption{
(a) A graph with the spine of a node $\nu$ shown in blue. We show the faces $(\first{\tour{\alpha}^\uparrow}, \last{\tour{\alpha}^\uparrow}, \first{\tour{\alpha}^\downarrow}, \last{\tour{\alpha}^\downarrow}$ in green.
(b) The cotree $\D$. Vertices are triangles.
  \vspace{-0.5cm}}
  \label{fig:designateddetection_large}
\end{figure}

\begin{observation}
Let $\nu$ be a path cluster with the slim-path property. Any edge in $G[\nu]$ is either an edge in $\T$ or it must intersect one of $\{ \tour{\nu}^\uparrow, \tour{\nu}^\downarrow \}$.
\end{observation}

\noindent
We introduce one last concept. Let $e = (x, y)$ be an edge of $G$ not in $\T$ where $e$ is an edge with one endpoint in $G[\alpha]$ in $G[\beta]$ (for two clusters $\alpha, \beta \in \TT$). We intuitively refer for each endpoint $x$ of $e$, to the tourpath intersected by $e$ 'near' $x$.  
Let $\alpha$ be a path cluster. We say that the endpoint $x \in \alpha$ of $e$ is a \emph{northern} endpoint if $e$ intersects $\tour{\alpha}^\uparrow$ near $x$, and a \emph{southern} endpoint if $e$ intersects $\tour{\alpha}^\downarrow$ near $x$. 
We define biconnected component borders (Figure~\ref{fig:border_large}):

\begin{definition}[Biconnected component \intervals]
Let $B$ be a subset of the edges in $G[\nu]$ that induces a biconnected subgraph such that $B \cap \pi(\nu)$ is a path (or singleton vertex). 
 \begin{itemize}[noitemsep, nolistsep]
     \item Let $\nu$ be a point cluster. Consider the clockwise ordering of edges in $B$ incident to its boundary vertex $u$, starting from $\des_u(\nu)$. The \textbf{\interval} of $B$ is:
     \begin{itemize}[noitemsep, nolistsep]
         \item the vertex $u$ together with its \textbf{eastern border}: the first edge of $B$ in this ordering, and its \textbf{western border}: the last edge of $B$ in this ordering.
     \end{itemize}
     \item Let $\nu$ be a path cluster. Denote by $a$ the `eastmost' vertex of $B \cap \pi(\nu)$ and by $b$ its `westmost' vertex. 
     If $a = b$ then the \interval is empty. Otherwise:
     \begin{itemize}[noitemsep, nolistsep]
         \item the \textbf{eastern border} of $B$ is $a$, together with the first northern and last southern edge of $B$ that is incident to $a$.
         \item the \textbf{western border} of of $B$ is $b$, together with the last northern and first southern edge of $B$ that is incident to $b$.
     \end{itemize}
 \end{itemize}
 \end{definition}
 
 \begin{figure}[h]
  \centering
  \includegraphics{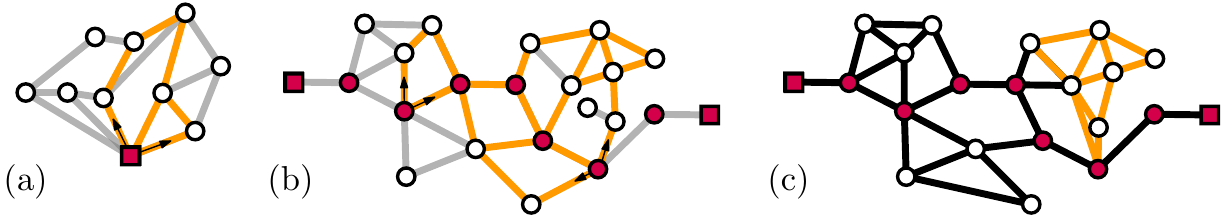}
  \caption{
  Three times a cluster $\nu$ where $\pi(\nu)$ are red vertices. 
  We show in yellow  a set of biconnected edges in $G[\nu]$ and show: (a) a border in a point cluster (b) a border in a path cluster and (c) a set of biconnected edges in $G[\nu]$ that has an empty border. 
  }
   \label{fig:border_large}
\end{figure}

\begin{restatable}{invariant}{storage}
\label{inv:component_storage}
For any $\nu \in \TT$ with boundary vertex $u$, for all $B \in BC_u^*(\nu)$ we store the \textbf{\intervals} of $B$ in $\nu$ and their indices and clockwise indices in $\nu$.
\end{restatable}
\paragraph{Using invariants for biconnectivity.}
Invariant~ \ref{inv:upwardpointers} allows us to not only obtain the meet between vertices, but also the edges of their path to this meet (incident to the meet):

\begin{restatable}{theorem}{meet}
\label{thm:meet}
Given Invariant~\ref{inv:upwardpointers}, let $\nu$ be a path cluster with boundary vertices $u$ and $v$ and $w\not\in\pi(\nu)$ be a vertex in $G[\nu]$. 
We can obtain the meet $m = meet(u, v, w)$ and the last  edge $e^*$ in the path from $w$ to $m$ (in $\T$) in $O(\log n)$ time.  
\end{restatable}

\begin{proof}
By Invariant~\ref{inv:upwardpointers-lca} we have a pointer from $w$ to the lowest common ancestor $\mu$ of the clusters that contain $w$. Starting from $\mu$ traverse $\TT$ upwards.  Each time the current cluster is an end merge of a point cluster $\alpha$ and a path cluster $\beta$, make a note of the boundary node $m$ and the unique incident edge $e^*$, which is stored for $\beta$ by Invariant~\ref{inv:upwardpointers-path}. When the traversal reaches $\nu$, the last $m$ and $e^*$ noted are the correct values.  Since this traverses a single root path in $\TT$ this takes $O(\log n)$ time.
\end{proof}

 In our later analysis, we show that for each merge, the only edges $e^\circ$ which can be part of new relevant and alive biconnected components are the edges incident to some convenient meets in the dual graph.
 Given such an edge $e^\circ$, we identify a convenient edge $e^*$ of the newly formed biconnected component $B$. We identify the already stored biconnected components $B^* \in BC_u(\nu)$ which contain $e^*$ (these components $B^*$ get `absorbed' into $B$). We use Invariant~\ref{inv:component_storage} to identify all such $B^*$ that contain $e^*$:

\begin{restatable}{theorem}{boundarycontainment}
\label{thm:boundarycontainment}
Let $e^* \in \T$ be an edge incident to a vertex $u$. 
Let $k$ be the maximum over all $u$ and $\nu$ of the number of elements in $BC^*_u(\nu)$.
In $O(k \log n)$ total time we can, for each of the $O(\log n)$ nodes $\nu \in \TT$ that contain $e^*$, for each $B^* \in BC_u^*(\nu)$,  determine if $e^*$ is in between the \interval of $B^*$ in $\nu$.
\end{restatable}

\begin{proof}
Here, $k$ is the maximal number of elements in $BC^*_u(\nu)$ for each $u$ and $\nu$.
    By Invariant~\ref{inv:component_storage} we know the relevant indices and clockwise indices clockwise of $B^*$, and by Lemma~\ref{lemma:indexing} and Lemma~\ref{lem:cwindex} we can compute the index and clockwise index of $e^*$ in all clusters containing them in $O(\log n)$ time, simultaneously.  For each of the $k$ elements $B\in BC^*_u(\nu)$ it now takes only a constant number of comparisons to determine if $e^*$ is in between the \interval of $B$ in $\nu$, so we can do this for all $O(k\log n)$ such $B$s in worst case $O(k\log n)$ total time.
\end{proof}

 \section{Summary of the remainder of this paper}
 We dynamically maintain a (combinatorial) embedding of some edge-divided graph $G$. 
 We maintain the top tree $\TT$ by Holm and Rotenberg from~\cite{holm2017dynamic} augmented with three invariants.
 All update operations on  the combinatorial embedding in Theorem~\ref{thm:main_full} can be realized by $O(\log n)$ split and merge operations on the top tree (and co-tree). We show how to maintain all three invariants with $O(\log n)$ additional time per merge as follows:

We define $k$ as the maximum over all vertices $u$ and clusters $\nu 
$, of the size of $BC^*_u(\nu)$. 
During each split, a cluster $\nu$ with boundary vertex $u$ is destroyed and we simply delete $BC^*_u(\nu)$ in $O(k)$ time.
Suppose we merge clusters $\alpha$ and $\beta$ to create a cluster $\nu$ with boundary vertex $u$. Any $B \in BC^*_u(\nu)$ contains at least one edge $e^\circ$ in $G[\alpha \cup \beta] \backslash (G[\alpha] \cup G[\beta])$. 

Section~\ref{sec:datastructure_proofs} specifies a set of invariants and theorems that allows us to navigate $\TT$. In addition, Theorem~\ref{thm:boundarycontainment} allows us to test for any edge $e^* \in T$ identify all pre-stored biconnected components $B^*$ that contain $e^*$ in $O(k \log n)$ \emph{total time}. This serves as our toolbox.

The edge $e^\circ$ must be part of some new biconnected component $B$ in $G[\nu]$.
Indeed: $e^\circ$ has one endpoint $x$ in $G[\alpha]$ and $y$ in $G[\beta]$. 
The edge $e^\circ$ together with the path $\pi^\circ$ in the spanning tree connecting $x$ and $y$ must form a cycle. 
We can test in $O(\log n)$ time whether $\pi^\circ$ is incident to the spine $\pi(\nu)$ and thus whether $B$ is relevant.

In Section~\ref{sub:intersection}, we show in Theorem~\ref{thm:componentcontainment}
that we can test whether $B$ is alive (i.e. incident to the face $\des_u(\nu)$). The core idea of this proof is (Figure~\ref{fig:overview} (a)) that
if $B$ is alive then: (1) $e^\circ$ or $\pi^\circ$ incident to $\des_u(\nu)$, or (2) $B$ contains some pre-stored biconnected $B^*$ incident to $\des_u(\nu)$.
We test case (1) with conventional methods in $O(\log n)$ time.
For case (2) we identify an edge $e^*$ on $\pi^\circ$ where if such a $B^*$ exists then $e^* \in B^*$.
We then apply Theorem~\ref{thm:boundarycontainment}.

In Section~\ref{sub:forming}, we show in Theorem~\ref{thm:merging} that for any such $e^\circ$, we can compute $B$ and its border in $G[\nu]$ in $O(k \log n)$ time. 
The core idea (Figure~\ref{fig:overview} (b)) is that when merging $\alpha$ and $\beta$ we can `project' $e^\circ$ onto $G[\alpha]$ and $G[\beta]$ to find the border of $B$ in the respective graphs. 
However, two complications arise: 
firstly, there may be some other edge $e^* \in G[\alpha \cup \beta] \backslash (G[\alpha] \cup G[\beta])$ which is also part of $B$. 
When we project $e^*$ onto $G[\alpha]$ and $G[\beta]$ we may reach `further' than $e^\circ$ (thus, expanding the border). Secondly, a merge can contain up to four clusters, not only two. We perform a case analysis where we show that we can construct the border of $B$ in $G[\nu]$ by pairwise joining projected borders. 

Finally in Section~\ref{sub:finalargument} we prove Theorem~\ref{thm:main_full}. 
We show that we do not need to consider `any' edge $e^\circ \in G[\alpha \cup \beta] \backslash (G[\alpha] \cup G[\beta])$.
Instead, we observe that any $e^\circ$ must intersect a tourpath of $\alpha$ and a tourpath of $\beta$.
For any fixed pair of tourpaths $(\tour{\alpha}, \tour{\beta})$, the edges in $G$ intersecting both tourpaths must lie on the meet between the Euler tours bounding this tourpath (Figure~\ref{fig:eulertours_huge}). 
This concept is similar to the \emph{edge bundles} by Laporte et al. in \cite{ItalianoPR93}).  We can restrict our attention to the first edge $e^\star$ of this bundle: the cycle between $\T$ and $e^\star$ encloses all other edges of the bundle in a face (thus, they cannot be part of some different alive biconnected component). 
By Theorem~\ref{thm:meet}, we can obtain $e^\star$ in $O(\log n)$ time. 
For each of these $\Theta(k)$ `maximal' edges $e^\star$, we apply the previous theorems to identify their relevant and alive biconnected components in $G[\nu]$ in $O(\log n)$ time to add them to $BC^*_u(\nu)$.

What remains is to upper bound the number of maximal edges and thus the integer $k$. 
A merge involves at most six different tourpaths. Thus, there are at most six choose two such interesting edges to form `new' biconnected components. This upper bounds the integer $k$ by $15$. 
We show that this allows us to maintain all invariants in $O(\log n)$ time per split and merge. 
To answer biconnectivity queries between $u$ and $v$, we expose $u$ and $v$ in $O(\log^2 n)$ time and use Invariant~\ref{inv:relevant_bicomp} to check for biconnectivity in $O(\log n)$ additional time. 

There exists one additional complication:
during a four-way merge, whenever $\alpha$ and $\beta$ are path clusters around a central vertex $m$, there exists no such `maximal' edge $e^\star$ (Figure~\ref{fig:overview} (c)). Thus, we cannot identify the biconnected components created by the edge bundle between $G[\alpha]$ and $G[\beta]$.
We observe that any such component is only useful if it connects the edges $e_1$ and $e_2$ of $\pi(\nu)$ incident to $m$. 
We test if removing $m$, separates $e_1$ and $e_2$ in $G$. 
This would be possible in $O(\log^2 n)$ using~\cite{holm2017dynamic} by splitting $m$ along the right corners and testing for connectivity. However, we want $O(\log n)$ update time per merge.
In Section~\ref{sec:specialcase} we open their black box slightly to test this in $O(\log n)$ time instead.

These proofs together show that we can dynamically maintain a combinatorial embedding subject to a broad set of operations and biconnectivity queries in $O(\log^2 n)$ time. 

\begin{figure}[h]
  \centering
  \includegraphics{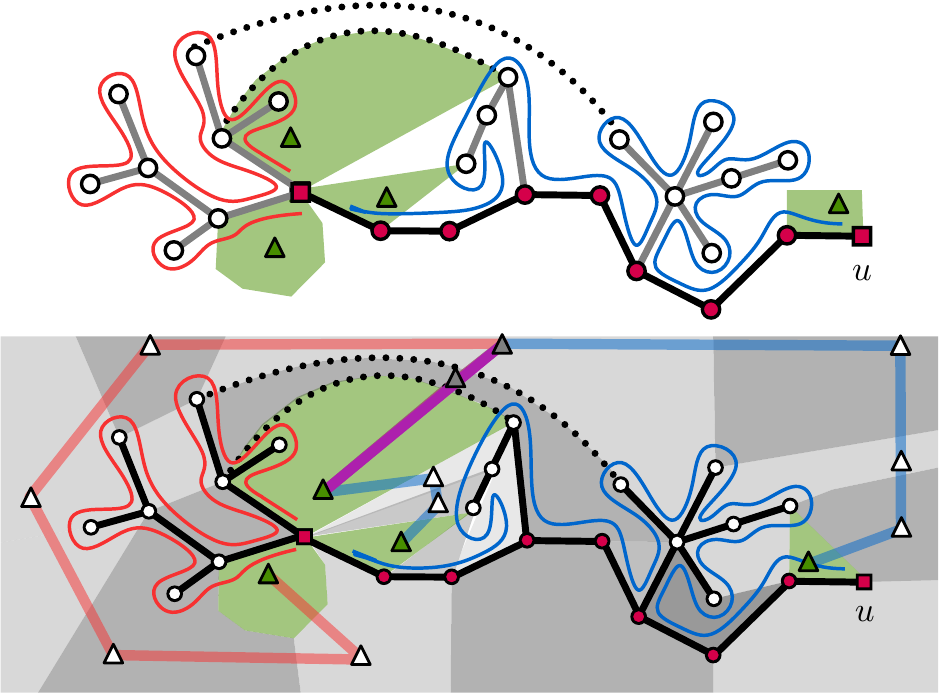}
  \caption{
An End merge between a point cluster $\alpha$ and a path cluster $\beta$ to create a new path cluster $\nu$.
We show two Euler tours $\tour{\alpha}$ and $\tour{\beta}^\uparrow$ in blue and red. 
The tour $\tour{\alpha}$ corresponds to the red path in the dual between two faces.
The tour $\tour{\beta}^\uparrow$ to the blue path. 
The purple path is their meet. 
Any edge $e^\circ \in G[\alpha \cup \beta] \backslash (G[\alpha] \cup G[\beta])$ intersects both $\tour{\alpha}$ and $\tour{\beta}^\uparrow$ (or $\tour{\alpha}$ and $\tour{\beta}^\downarrow$) and must thus lie on the purple path (or an alternative meet in the dual). 
The first edge on this path is $e^\star$, as any further edge cannot be incident to the face $\des_u(\nu)$.
  }
  \label{fig:eulertours_huge}
\end{figure}

\begin{figure}[hp]
  \centering
  \includegraphics{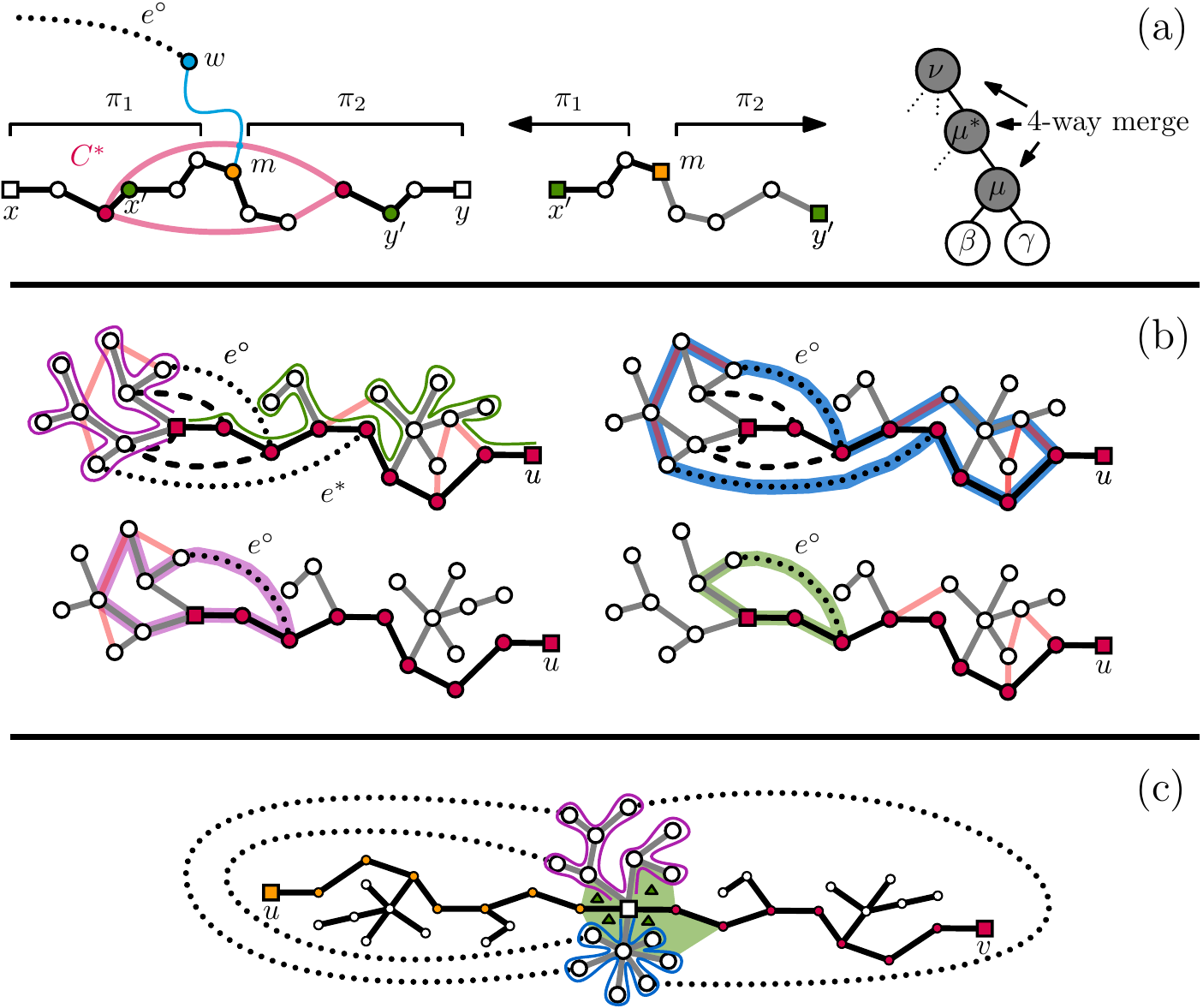}
  \caption{ Three challenges described in our overview.
  \newline(a) Let $e^\circ$ have some endpoint $w$ and consider the path in $\T$ to some vertex $m$. Let $B$ be the biconnected component formed by $e^\circ$. 
  There exists some child $\mu$ of $\nu$ where $m$ is the central vertex of the merge. If $B$ contains some pre-stored biconnected component $B^*$ (red) then $B$ includes either the edge $e_1$ or $e_2$ in $G[\mu]$ incident to $m$. \newline
  (b) Consider an End Merge and an edge $e^\circ$ intersecting the purple and green Euler tours. 
  The edge $e^\circ$ is part of a biconnected component with the blue cycle as its outer cycle. 
  We find for $e^\circ$ however, only the purple and green cycles in $G[\alpha]$ and $G[\beta]$. We smartly join these cycles together with the cycles for $e^*$ to get $B^*$. \newline
  (c) In a four-way merge, the edges incident to the outer face of the embedding may be arbitrary edges in the edge bundle between the two path clusters. Neither edges incident to the outer face are incident to $\des_u(\nu)$ or $\des_v(\nu)$. Since we have no techniques for finding these edges, we instead test if the central vertex separates $(u, v)$.
  } 
  \label{fig:overview}
\end{figure}

\section{Dynamic maintenance of our augmented top tree}

We aim to prove the following theorem:

\theoremmain*

Let $\nu \in \TT$ be a (path) cluster with children $\mu_1, \mu_2, \ldots \mu_s$, and denote by $u$ (and $v$) its boundary vertex (vertices). 
To prove Theorem~\ref{thm:main_full}, we want to identify all biconnected components $B$ in $G[\nu]$ that are not a biconnected component in $G[\mu_i]$ for some $i$ (where $B$ is relevant in $\nu$ and alive with respect to $\des_u(\nu)$ or $\des_v(\nu)$).  We observe the following:

\begin{observation}
\label{obs:pairwise}
Let $\nu$ be a node with children $\mu_1, \mu_2, \ldots \mu_s$. 
Then $G[\nu] = \bigcup_{i, j} G[\mu_i \cup \mu_j]$.
\end{observation}

\noindent
Any biconnected component $B$ that exists in $G[\nu]$ but not in $\bigcup_i G[\mu_i]$ either is the union of biconnected components in $G[\mu_i]$, or contains an edge $e^\circ$ in $G[\nu] \setminus \bigcup_i G[\mu_i]$.
By Observation~\ref{obs:pairwise}, it suffices to check for every pair $(\mu_i, \mu_j)$  the edges $e$ in $G[\mu_i \cup \mu_j] \setminus (G[\mu_i] \cup G[\mu_j])$ and the biconnected components that these edges $e$ create. 

\subsection{
\texorpdfstring{
An edge $e^\circ$ in $G[\nu] \setminus \bigcup_i G[\mu_i]$ and whether it is part of an alive BC}{An edge e in G[v] \ U G[m(i)] and whether it is part of an alive BC}
}
\label{sub:intersection}
The above observations inspire us to consider the following setting: let $\nu$ be a cluster with at least two children $\alpha$ and $\beta$, and $e^\circ$ be an edge where one endpoint is in $\alpha$ and one endpoint is in $\beta$. 
We show in the proof of our main theorem that the edge $e^\circ$ is part of some \emph{relevant} biconnected component $B_{\alpha \beta}$ of $G[\alpha \cup \beta]$. What then remains, is to identify whether $B_{\alpha \beta}$ is edge-incident to 
$\des_u(\nu)$ (i.e. is $B_{\alpha \beta}$ also an \emph{alive} biconnected component of $G[\nu]$). Suppose that $B_{\alpha \beta}$ is indeed edge-incident to $\des_u(\nu)$ then either: 
\begin{enumerate}[noitemsep, nolistsep]
    \item the edge $e^\circ$ is edge-incident to $\des_u(\nu)$, or
    \item $e^\circ$ is biconnected to the edges in a biconnected component $B^*$ in $G[\alpha]$ (or $G[\beta]$) that is edge-incident to $\des_u(\nu)$.
\end{enumerate}

\begin{restatable}{theorem}{edgeincident}
\label{thm:edgeincident}
Given the cotree $\D$ we can, for any cluster $\nu$ with boundary vertex $u$ and any edge $e^\circ \in G[\nu] \setminus \T$, decide if $e^\circ$ is edge-incident to $\des_u(\nu)$ in $O(\log n)$ time. 
\end{restatable}

\begin{proof}
Per definition, $e^\circ$ is an edge in $\D$. The cluster $\nu$ is either a path or point cluster.

Let $\nu$ be a point cluster and consider the tourpath $\tour{\nu}$ and the  faces $\first{\tour{\nu}}$ and $\last{\tour{\nu}}$ that intersect this tour path. The edge $e^\circ$ is edge-incident to $\des_u(\nu)$ if and only if it is vertex-incident to the path $\pi$ from  $\first{\tour{\nu}}$ to $\last{\tour{\nu}}$ in the cotree $\D$. Using Holm and Rotenberg~\cite{holm2017dynamic} we can, given a pointer to $e^\circ$, detect this in $O(\log n)$ time. 

Let $\nu$ be a path cluster and consider the tourpaths $\tour{\nu}^\uparrow$ and $\tour{\nu}^\downarrow$. Let without loss of generality $\des_u(\nu)$ be incident to the start of $\tour{\nu}^\uparrow$ and the end of $\tour{\nu}^\downarrow$. 
The edge $e^\circ$ is edge-incident to $\des_u(\nu)$ if and only if it is vertex-incident to the path $\pi$ from  $\first{\tour{\nu}^\uparrow}$ to $\last{\tour{\nu}^\downarrow}$ in the cotree $\D$. Again, we can detect this in $O(\log n)$ time using~\cite{holm2017dynamic}.
\end{proof}

\noindent 
Theorem~\ref{thm:edgeincident} checks the first condition.
Theorem~\ref{thm:componentcontainment} (which we prove later in this subsection) checks the second:

\begin{restatable}{theorem}{componentcontainment}
\label{thm:componentcontainment}
Let $\nu \in \TT$, $\alpha$ and $\beta$ be two children of $\nu$, and $u$ be a boundary vertex of $\nu$. Let $e^\circ$ be an edge with one endpoint in $\alpha$ and one endpoint in $\beta$. Given Invariants~\ref{inv:relevant_bicomp}, \ref{inv:upwardpointers} and \ref{inv:component_storage}, we can identify  the biconnected component $B^* \in BC(\alpha, \proj_u(\alpha))$ where $e^\circ$ is biconnected to the edges of $B^*$ in the graph $G[\alpha] \cup \T$ in $O(k \log n)$ time (or conclude  no such $B^*$ exists).
Here, $k$ is the maximum over all $u$ and $\nu$ of the number of elements in $BC^*_u(\nu)$.
\end{restatable}

\noindent
Our approach towards proving this theorem is as follows: an edge $e^\circ \in G[\alpha \cup \beta]$ is attached to some vertex $w \in G[\alpha]$. Either:
\begin{enumerate} [noitemsep, nolistsep]
    \item $\alpha$ is a path cluster. Then this vertex $w$ is connected (in the spanning tree $\T$) to some internal vertex $m \in \pi(\nu)$.  Either the removal of $m$ separates $\pi(\alpha)$ in $G[\alpha]$, or, $m$ is encapsulated by some biconnected component $B^* \in BC(\alpha, \proj_u(\alpha))$. Else,
    \item $\alpha$ is a point cluster. Let $w$ be connected to the boundary vertex by some edge $e^* \in \T$. Either $e^*$ is in some biconnected component $B^* \in BC(\alpha, \proj_u(\alpha))$, or $e^\circ$ is not part of a biconnected component in $BC(\alpha, \proj_u(\alpha))$.
\end{enumerate}

\begin{lemma}[Case 1: $\alpha$ is a path cluster]
\label{lemma:pathdetection}
Let $\alpha \in \TT$ be a path cluster with boundary vertices $x$ and $y$.
Let $m$ be an internal vertex on the path $\pi(\alpha)$ such that:
$m$ is path-connected to a vertex $w$ that is incident to $\des_x(\alpha)$ in the graph $G[\alpha] \setminus \pi(\alpha)$ (Figure~\ref{fig:articulationdetection}). 
Either: 
\begin{itemize}[noitemsep, nolistsep]
\item there exists a $B^* \in BC_x(\nu, \des_x(\alpha))$ that contains both edges of $\pi(\alpha)$ incident to $m$,
    \item or removing the vertex $m$ separates the spine $\pi(\alpha)$ in $G[\alpha]$.
\end{itemize}
 Given Invariants \ref{inv:relevant_bicomp}, \ref{inv:upwardpointers} and \ref{inv:component_storage} we can identify $B^*$ (if it exists) in $O(k \log n)$ time. 
\end{lemma}

\begin{figure}[b]
  \centering
  \includegraphics{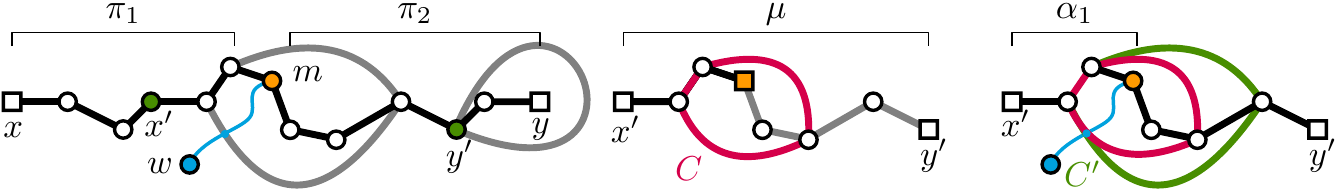}
  \caption{
 The spine of a path cluster $\alpha$ and some vertex $m$.
  The vertex $w$ and the path in $G[\alpha]\setminus \pi(\alpha)$ is shown in blue. There exists a unique 4-way merge between path clusters (black and grey) where $m$ is the central vertex. 
  If $m$ does not separate $(x', y')$ in $G[\alpha]$, it must be contained in some cycle $C$. Thus, there exists some cycle $C'$ bordering $\des_x(\alpha)$.
  }
  \label{fig:articulationdetection}
\end{figure}

\begin{proof}
The proof is illustrated by Figure~\ref{fig:articulationdetection}.
Denote by $\pi^*$ the path from $m$ to $w$ in $G[\alpha] \setminus \pi(\alpha)$.
Denote by say $\pi_1$ is the subpath of $\pi(\alpha)$ from $u$ to $m$  (excluding $m$). Denote by  $\pi_2$ is the subpath from $\pi(\alpha)$ from $m$ to $v$ (excluding $m$).
Suppose that removing $m$ does not separate $\pi_1$ from $\pi_2$ in $G[\alpha]$. 
First, we show that it must be that there exists a  $B^* \in BC_u(\alpha, \des_x\alpha))$ that contains both edges of $\pi(\alpha)$ incident to $m$.
Indeed, if removing $m$ does not separate $\pi_1$ and $\pi_2$ (in $G[\alpha]$) then there exists a cycle $C$ in $G[\alpha]$ intersects both a vertex of $\pi_1$ and a vertex of $\pi_2$.  Let us call such a cycle a \emph{witness}.
We claim that, since $m$ is path-connected to a vertex $w$ that is incident to $\des_x(\alpha)$ in the graph $G[\alpha] \setminus \pi(\alpha)$, there must exist some witness cycle $C' \in G[\alpha]$ that is edge-incident to $\des_x(\alpha)$. Indeed, suppose that the witness $C$ is not edge-incident to $\des_x(\alpha)$. Then the witness it must be incapsulated in some cycle $C'$ that is edge-incident to $\des_x(\alpha)$.
This cycle $C'$ must intersect the path $\pi^*$. However, this implies that there is at least one witness cycle $C''$ that is edge-incident to $\des_x(\alpha)$ (this cycle is obtained by combining a subpath of $\pi^*$ with a subpath of $C'$ that is edge-incident to $\des_x(\alpha)$). 
Denote by $C^*$ the largest witness cycle that is edge-incident to $\des_x(\alpha)$.
The cycle $C^*$ must be contained in  some biconnected component $B^* \in BC_x(\alpha, \des_x(\alpha))$.
This proves our claim that removing $m$ does not separate $\pi(\alpha)$ in $G[\alpha]$  if and only if a biconnencted component $B^*$ that contains both edges of $\pi(\alpha)$ incident to $m$ exists.
What remains, is to show that we can identify $B^*$ in $O(k \log n)$ time.

By Lemma~\ref{lemma:unique_descendent}, there exists a unique descendant $\mu^*$ of $\alpha$ where $B^* \in BC^*_{x^*}(\mu^*)$ (and $x^*$ is the boundary vertex of $\mu^*$ that is closest to $x$). 
We show that $\mu^*$ can be only one out of $O(\log n)$ descendants of $\nu$ (Figure~\ref{fig:articulationtree}). 
Per definition, $B^* \cap \pi(\alpha)$ has one endpoint in $\pi_1$ and one endpoint in $\pi_2$.
Consider the unique four-way merge where $m$ is the central vertex, which merged two path clusters $\beta$ and $\gamma$ to create some cluster $\mu \in \TT(\alpha)$ with boundary vertices $x'$ and $y'$.
The cluster $\mu$ is a descendant of $\alpha$. The path in $\TT$ from $\mu$ to $\alpha$ must consist of only four-way merges (any other merge results in a node that is not a path cluster, and all ancestors of such a node cannot have the vertex $m$ on their spine).
Any descendant $\mu'$ of $\alpha$ where $B^* \in BC_x(\mu', \proj_x(\mu'))$ must be on this path in $\TT$ (since any further descendants contain only a part of $\pi_1$ \emph{or} $\pi_2$. 
We check all the $O(\log n)$ four-way merges on this path, starting from the four-way merge that created $\alpha$. 

\begin{figure}[h]
  \centering
  \includegraphics{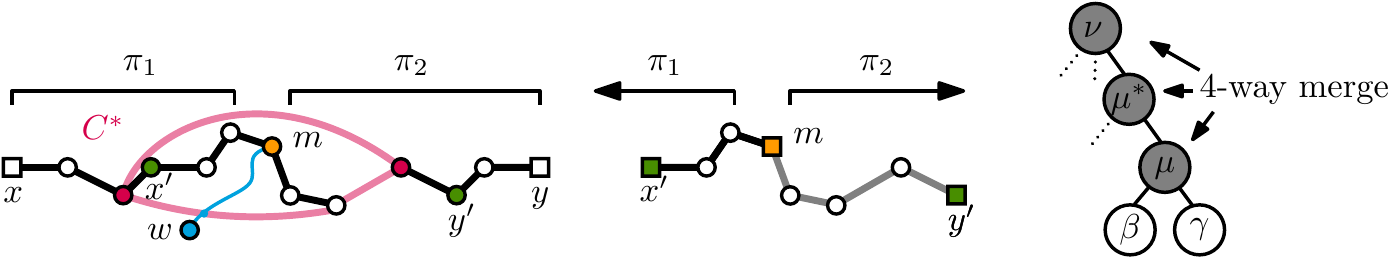}
  \caption{
 The cluster $\alpha$ and the maximal cycle $C^*$ that is part of the biconnected component $B^*$. We show the cluster $\mu$ that is the result of the four-way merge where $m$ is the central vertex.
 The child clusters $\beta$ and $\gamma$ of $\mu$ coincide with $\pi_1$ and $\pi_2$ respectively.
 The node $\mu^*$ must be contained in the path in $\TT$ between $\mu$ and $\alpha$.
  }
  \label{fig:articulationtree}
\end{figure}

Let $\mu^*$ be any path cluster on the aforementioned path in $\TT$.
By Lemma~\ref{lemma:indexing}, we can obtain the index of $m$ in each $\mu^*$ in $O(\log n)$ total time. 
We can obtain the boundary vertex $x^*$ of $\mu^*$ that is closest to $u$ in $O(1)$ time.
Each biconnected component $B \in BC^*_{x^*}(\mu^*)$ that may contain both spine edges incident to $m$ must be a \interval of type $2$. 
Thus, by Invariant~\ref{inv:component_storage}, we have a pointer to the endpoints $a_1$ and $a_2$ of $B \cap \pi(\mu)$ and their indices in $\mu$.
We can check if $m$ is in between $a_1$ and $a_2$ in $O(\log n)$ time.
The edges incident to $m$ are in $B$ if and only if this is the case.  Thus, by iterating over the at most $k$ elements in $BC^*_{x^*}(\mu^*)$ we can test if $B^* \in BC^*_{x^*}(\mu^*)$ in $O(\log n)$ time.
The lemma follows. 
\end{proof}

\begin{lemma}[Case 2: $\alpha$ is a point cluster]
\label{lemma:ptcltrcontainment}
Let $\alpha \in \TT$ be a point cluster with boundary vertex $x$. Let $e^*$ be an edge with the following properties (Figure~\ref{fig:pathclusterdetection}):
\begin{itemize}[noitemsep, nolistsep]
    \item $e^*$ is in $\T$ and in $G[\alpha]$, and
    \item $e^*$ be path-connected in $G[\alpha \setminus \{ x \} ]$ to a vertex $w \neq u$ which is incident to $\des_x(\alpha)$.
\end{itemize}  
Given Invariants~\ref{inv:relevant_bicomp}, \ref{inv:upwardpointers} and \ref{inv:component_storage}, we can compute the biconnected component $B^* \in BC_u(\nu, \des_u(\nu))$ that contains $e^*$ in $O(k \log n)$ time or conclude that no such $B^*$ exists.
\end{lemma}

\begin{proof}
The proof is illustrated by Figure~\ref{fig:pathclusterdetection}.
By Lemma~\ref{lemma:unique_descendent}, there exists a unique descendant $\mu$ of $\alpha$ where $B^* \in BC^*_{x^*}(\mu)$ (and $x^*$ is the closest vertex to $x$ in $\T$).
Per definition, $e^*$ is in $T$. So denote by $\mu_1$ the leaf in $\TT$ containing only $e^*$ and by $\mu_1, \mu_2, \ldots \alpha$ the root-to-leaf subpath of clusters where $G[\mu_i]$ contains $e^*$.
There are at most $O(\log n)$ such nodes. Each $\mu_i$ (that is a descendant of $\alpha$) must have $x$ as one of its boundary vertices (since $e^*$ is incident to $x$).
For any index $i$, we denote by $B_i \in BC_{x}^*(\mu_i)$ the biconnected component that contains $e^*$ (if it exists).
Observe that since biconnectivity is an equivalence relation over the edges in the set, it must be that $B_i \subset B_j$ for all $i$ and $j$  with $i < j$ (should $B_i$ and $B_j$ both exist). It follows that $B^* = B_j$ where $j$ is the largest index $j$ for which $B_j$ exists. 
Thus, we identify $B^*$ by checking in a top-down fashion for each $\mu_j$ whether there exists a biconnected component $B_j \in BC_{x}^*(\mu_j)$ that contains $e^*$. We investigate two cases:

\textbf{Case 1: $\mu_j$ is a point cluster.}
Let $\mu_j$ be a point cluster with boundary vertex $x$. By Invariant~\ref{inv:component_storage}, we have for each $B' \in BC_{x}^*(\mu_j)$ a pointer to the \interval of $B'$. In $O(k \log n)$ time, we check if $e^*$ is contained between the eastern edge and western edge (Theorem~\ref{thm:boundarycontainment}). Let $e^*$ be contained in the \interval formed by $e_1$ and $e_2$.
Since $e^*$ is path connected in $G[\alpha] \setminus \pi(\alpha)$ to a vertex $w$ that is incident to $\des_x(\alpha)$, it follows that $e^* \in B'$.
Suppose otherwise that $e^*$ is not in between $e_1$ and $e_2$, then per definition $e^*$ cannot be in $B'$. 
We check for every $j$, every $B' \in BC_u^*(\mu_j)$ in $O(k \log n)$ time, which takes $O(k \log n)$ total time.

\textbf{Case 2: $\mu_j$ is a path cluster.}
Let $\mu_j$ be a path cluster where $x$ is a boundary vertex. Since $e^*$ is incident to $x$ and we are maintaining a slim-path top tree, it must be that $e^*$ is the only edge incident to $x$ in $G[\mu_j]$ and thus, $e^*$ cannot be part of a biconnected component in $G[\mu_j]$ and we elect to skip over $\mu_j$ in $O(1)$ time.
\end{proof}

\begin{figure}[t]
  \centering
  \includegraphics{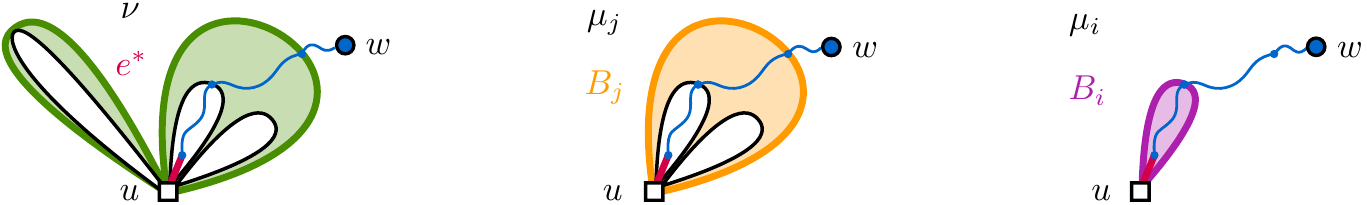}
  \caption{
  A path cluster $\alpha$ with boundary vertex $x$. We consider an edge $e^* \in \T$ that is incident to $x$ and path connected (in $G[\alpha \setminus \{x \} ]$ to a vertex $w$ that is incident to the face $\des_x(\alpha)$.
  \vspace{-0.5cm}}
  \label{fig:pathclusterdetection}
\end{figure}

\componentcontainment*

\begin{proof}
Denote by $x$ the boundary vertex of $\alpha$ that is closest to $u$ (in the spanning tree $\T$).

\textbf{ Let $\alpha$ be a point cluster.} 
If $e^\circ$ is incident to $x$ then $e^\circ$ cannot be biconnected to the edges of any biconnected component of $G[\alpha]$. We can test if $e^\circ$ is incident to $x$ in $O(1)$ time.
Suppose that $e^\circ$ is not incident to $x$. 
Denote by $w$ the endpoint of $e^\circ$ in $G[\alpha]$ and consider the path $\pi^*$ in $\T$ from $w$ to $x$. We denote by $e^*$ its last edge. 
The edge $e^\circ$  is biconnected to the edges in a set $B^* \in  BC(\alpha, \proj_u(\alpha))$ if and only if $e^* \in B^*$ (indeed, for any such $B^*$ the path $\pi^*$ connects $e^*$ to two vertices of $B^*$). 
We immediately apply Lemma~\ref{lemma:ptcltrcontainment} and identify if $B^*$ exists in $O(k \log n)$ time.

\textbf{Let $\alpha$ be a path cluster.}
First, we consider the special case where $w$ is a boundary vertex of $\pi(\alpha)$.
Denote by $e^*$ the edge of $\pi(\alpha)$ at distance $1$ of $w$ (i.e. not the edge $e_1$ incident to $w$, but the edge incident to $e^*$). 
Observe that if the edge $e^\circ$ is biconnected in $G[\alpha] \cup \T$ to the edges of some $B^* \in BC(\alpha, \proj_u(\alpha))$ then $B^*$ must contain $e^*$. Analogue to the proof of Lemma~\ref{lemma:pathdetection}, we can test if there exists a $B^* \in BC(\alpha, \proj_u(\alpha))$ that contains $e^*$ in $O(k \log n)$ time. 

Now consider the canonical case where $w$ is not a boundary vertex of $\pi(\alpha)$. denote by $m = meet(x, y, w)$. By our definition of slim-path top trees, $m$ must be an internal vertex on $\pi(\alpha)$. We note that either:

\begin{enumerate}[noitemsep, nolistsep]
\item there exists a $B^\circ \in BC_x(\nu, \des_x(\alpha))$ that contains both edges of $\pi(\alpha)$ incident to $m$
    \item or removing the vertex $m$ separates the spine $\pi(\alpha)$ in $G[\alpha]$.
\end{enumerate}
By Lemma~\ref{lemma:pathdetection}, we can identify $B^\circ$ (if it exists) in $O(k \log n)$ time. In the case where $B^\circ$ exists (Case 1), $e^\circ$ is biconnected to the edges in $B^\circ$ in the graph $G[\nu]$ (this follows from the observation that $e^\circ$ is path-connected to $\pi(\alpha)$ in $G[\alpha]$ and to $\pi(\beta)$ in $G[\beta]$). 

Suppose (Case 2) that $B^\circ$ does not exist and suppose that there \emph{does} exist a biconnected component $B^* \in BC(\alpha, \proj(\alpha))$ whose edges are biconnected to $e^\circ$ in $G[\nu]$.
We show that we can identify this component $B^*$ with a further case distinction. Denote by $\mu$ the four-way merge where $m$ was the central vertex. We have a pointer to $\mu$ through Invariant~\ref{inv:upwardpointers}.

\noindent
In this scenario, it must be that either:
\begin{enumerate}[noitemsep, nolistsep]
    \item the biconnected component $B^*$ contains an edge of $\pi(\alpha)$ that is incident to $m$, or
    \item the biconnected component $B^*$ is in $BC(\mu, \proj_u(\mu))$.
\end{enumerate}

Case 2.1: There are at most $O(\log n)$ nodes $\mu'$ in $\TT$ that contain a spine edge incident to $m$. We can check for each of them if there exists a $B' \in BC_m^*(\mu)$ with $e^* \in B^*$ in $O(k \log n)$ total time (check for each $B'$ if $e^*$ is contained in the \interval of $B'$ using Theorem~\ref{thm:boundarycontainment}). 

Case 2.2: We handle this case analogously to the case where $\alpha$ is a path cluster. 
\end{proof}

\subsection{
\texorpdfstring{An edge $e^\circ$ in $G[\nu] \setminus \bigcup_i G[\mu_i]$ and the BC it forms in $G[\nu]$}{An edge e in G[v] \ U G[m(i)] and the BC it forms in G[v]}
}
\label{sub:forming}

The previous subsection and its running times relied on some integer $k$: the maximum over all $u$ and $\nu$ of the number of elements in $BC^*_u(\nu)$.
Before we show that the integer $k$ is upper bounded by a constant, we first establish the following technical result. 
Suppose that $\nu$ is a path cluster with children $\alpha$ and $\beta$, and $e^\circ$ is an edge in $G[\alpha \cup \beta]$ that is not in either $G[\alpha]$ or $G[\beta]$. Then we \emph{know} that $e^\circ$ is part of some biconnected component $B_{\alpha \beta}$ in the graph $G[\alpha \cup \beta]$. 
Theorem~\ref{thm:merging} aims to identify for $B_{\alpha \beta}$ the `important' information of $B_{\alpha \beta}$. Slightly more formally, we show how we can compute what we later call its \emph{projected component}: the biconnected component formed by $e^\circ$, all edges in $G[\beta]$ and all spanning tree edges in $G[\alpha]$ (see Figure~\ref{fig:projectedcomponent}).

\begin{restatable}{theorem}{merging}
\label{thm:merging}
Let $\nu \in \TT$, $\alpha$ and $\beta$ be two children of $\nu$.
Let $e^\circ \in G[\alpha \cup \beta] \setminus ( G[\alpha] \cup G[\beta])$. Then $e^\circ$ is part of a biconnected component $B$ in the graph $\{ e^\circ \} \cup G[\beta] \cup \left ( \T \cap G[\alpha] \right)$.
Moreover, we can identify in $O(k \log n)$ time:
\begin{enumerate}[(a), noitemsep, nolistsep]
    \item the path $B \cap \pi(\beta)$, 
    \item whether $B$ is \emph{relevant} in $G[\nu]$, and
    \item the \intervals of $B$ in $G[\nu]$.
\end{enumerate}
Here, $k$ is the maximum over all $u$ and $\nu$ of the number of elements in $BC^*_u(\nu)$.
\end{restatable}

\begin{proof}[Proof of Theorem~\ref{thm:merging}]
First, we note that any edge $e^\circ \in G[\alpha \cup \beta] \setminus ( G[\alpha] \cup G[\beta])$ is indeed part of some biconnected component $B$ in the graph $\{ e^\circ \} \cup G[\beta] \cup \left ( \T \cap G[\alpha] \right)$. Indeed, one endpoint of $e^\circ$ is path-connected to $\pi(\beta)$ in $G[\beta] \cap \T$ and the other endpoint of $e$ is path-connected to $\pi(\alpha)$ in $G[\alpha] \cap \T$. These two disjoint paths create, together with $e^\circ$, a cycle in $\{ e^\circ \} \cup G[\beta] \cup \left ( \T \cap G[\alpha] \right)$. 
What remains is to compute the properties (a), (b) and (c) of this biconnected component $B$.  
The proof is an elaborate case distinction based on whether $\alpha, \beta$ and $\nu$ are path or point clusters. \newline

\noindent
\noindent
\textbf{Case 1: a Point merge.} 

\textbf{Case 1.1 $\nu$ and $\beta$ are point clusters that share a boundary vertex $u$.}
\begin{enumerate}[(a), noitemsep, nolistsep]
\item The biconnected component $B$ must intersect $\pi(\beta)$.
Since $\pi(\beta)$ is a single vertex, it follows that $B \cap \pi(\beta)$ is the spine of $\beta$.
\item Per definition, $\pi(\beta) = \pi(\nu)$ and so $B$ must be relevant in $G[\nu]$.
\item Let without loss of generality, $G[\alpha]$ precede $G[\beta]$ in the clockwise ordering around $u$. Denote by $z$ the endpoint of $e^\circ$ in $\alpha$ and by $e_1$ the last edge in the path from $z$ to $\pi(\alpha)$ in the spanning tree $e_1$. By Property~\ref{prop:ordering}, $e_1$ is the eastern \interval of $B$.  

Denote by $w$ the endpoint of $e^\circ$ in $\beta$ and by $e^*$ the last edge in the path from $w$ to $\pi(\beta)$. 
By Lemma~\ref{lemma:ptcltrcontainment}, we can determine if there exists a $B^* \in BC(\beta, \des_u(\beta))$ where $e^* \in B^*$. If such a $B^*$ exists, then the western \interval of $B$ is the western \interval of $B^*$. Otherwise, the western \interval of $B$ must be $e^*$ itself. 
\end{enumerate} \leavevmode 

\noindent
\textbf{Case 2: End merge.}

\textbf{Case 2.1: $\nu$ and $\beta$ are point clusters with $\pi(\nu) = u$, $\alpha$ is a path cluster.}
\begin{enumerate}[(a), noitemsep, nolistsep]
    \item The biconnected component $B$ must intersect $\pi(\beta)$.
Since $\pi(\beta)$ is a single vertex, it follows that $B \cap \pi(\beta)$ is the spine of $\beta$.
\item Per definition, $B$ is relevant in $G[\nu]$ if and only if it contains two edges $e_1, e_2$ that are incident to the boundary vertex $u$ of $\nu$. Moreover, $u$ has to be a boundary vertex of $\alpha$ and not $\beta$.
Via the slim-path property of $\alpha$, these two edges $e_1$ and $e_2$ cannot both be edges in the spanning tree $\T$. Thus, it follows that $B$ is relevant in $G[\nu]$ if and only if $e^\circ$ is incident to $u$. We can test this in $O(1)$ time. 
\item Let $B \cap \pi(\nu) \neq \emptyset$. Then the eastern \interval is the edge of $\pi(\alpha)$ that is incident to $u$ and the west \interval is $e^\circ$.
\end{enumerate}

\textbf{Case 2.2: $\nu$ and $\alpha$ are point clusters with $\pi(\nu) = u$, $\beta$ is a path cluster. }
\begin{enumerate}[(a), noitemsep, nolistsep]
    \item Denote by $w$ the endpoint of $e^\circ$ in $\beta$.
    Denote by $m$ the meet between $w$ and the boundary vertices of $\beta$ and by $e^*$ the edge incident to $m$ that is farthest away from $u$. 
    We obtain $m$ and $e^*$ in $O(\log n)$ time (Theorem~\ref{thm:meet}).
    The edge $e^*$ must be in $B$. 
    
    By Lemma~\ref{lemma:pathdetection}, we can detect if there exists a biconnected component $B^* \in BC(\beta, \des_u(\beta))$ in $O(k \log n)$ time. If such a biconnected component $B^*$ exists it must be unique. Moreover, since $B^*$ is a biconnected component of $G[\beta]$, $B^*$ must be a subset of $B$. 
    If $B^*$ exists, then the path $B \cap \pi(\beta)$ goes from the boundary vertex $v$ of $\beta$ to the endpoint of $B^* \cap \pi(\beta)$.
    If no such $B^*$ exists then the path $B \cap \pi(\beta)$ goes from $v$ to $m$. 
  \item We check if $u$ is in $B \cap \pi(\beta)$ in $O(1)$ additional time. 
  \item Let $B \cap \pi(\nu) \neq \emptyset$ and reconsider the case distinction in our argument for (a). If $B^*$ exists, then the \interval of $B$ in $G[\nu]$ is equal to the western \interval of $B^*$. Otherwise, the \interval of $B$ in $G[\nu]$ has as its eastern and western \interval $e^\circ$ and the edge $e$ of $\pi(\beta)$ that is incident to the boundary vertex $u$. We decide which edge is eastern and which is western in $O(1)$ additional time. 
\end{enumerate} \leavevmode

 \noindent
\textbf{Case 3: Four-way merge}

\textbf{Case 3.1: $\nu$ and $\alpha$ are path clusters and $\beta$ is a point cluster.}
\begin{enumerate}[(a), noitemsep, nolistsep]
    \item Since $\beta$ is a point cluster, the path $B \cap \pi(\beta)$ is the boundary vertex of $\beta$.
    \item Denote by $w$ the endpoint of $e^\circ$ in $G[\alpha]$ and by $m$ the meet between the boundary vertices of $\alpha$ and $m$. We obtain $m$ in $O(\log n)$ time (Theorem~\ref{thm:meet}). The path $B \cap \pi(\alpha)$ can, per definition, only use edges in $\T \cap G[\alpha]$. Thus, the path $B \cap \pi(\nu)$ is equal to the path from $m$ to the boundary vertex of $\beta$. 
\item The eastern \interval of $B \cap \pi(\nu)$ is the aforementioned vertex $m$ together with the two edges of $B$ that are incident to $m$. 

Computing the western \interval is slightly more involved. Let $x$ be the boundary vertex of $\beta$. Denote by $w$ the endpoint of $e^\circ$ in $\beta$ and by $e^*$ the last edge on the path in $\T$ from $w$ to $x$.
Using Lemma~\ref{lemma:ptcltrcontainment}, we test if $e^*$ is contained in a biconnected component $B^* \in BC_x(\beta, \des_x(\beta))$ in $O(k \log n)$ time. 
If such a $B^*$ exists, then the western \interval of $B$ in $\pi(\nu)$ is the vertex $x$; together with the edge of $\pi(\alpha)$ that is incident to $x$ and the western \interval of $B^*$. 
Otherwise, the western \interval of $B$ in $\pi(\nu)$ is the vertex $x$; together with the edge of $\pi(\alpha)$ that is incident to $x$ and the edge $e^*$. 
\end{enumerate}

\textbf{Case 3.2: $\nu$ and $\beta$ are path clusters and $\alpha$ is a point cluster.}
\begin{enumerate}[(a), noitemsep, nolistsep]
    \item Denote by $w$ the endpoint of $e^\circ$ in $\beta$.
    Denote by $m$ the meet between $w$ and the boundary vertices of $\beta$ and by $e^*$ the edge incident to $m$ that is farthest away from $u$. 
    We obtain $m$ and $e^*$ in $O(\log n)$ time (Theorem~\ref{thm:meet}).
    The edge $e^*$ must be in $B$. 
    
    By Lemma~\ref{lemma:pathdetection}, we can detect if there exists a biconnected component $B^* \in BC(\beta, \des_u(\beta))$ in $O(k \log n)$ time. If such a biconnected component $B^*$ exists it must be unique. Moreover, since $B^*$ is a biconnected component of $G[\beta]$, $B^*$ must be a subset of $B$. 
    If $B^*$ exists, then the path $B \cap \pi(\beta)$ goes from the boundary vertex $v$ of $\beta$ to the endpoint of $B^* \cap \pi(\beta)$.
    If no such $B^*$ exists then the path $B \cap \pi(\beta)$ goes from $v$ to $m$. 
    \item Observe that $\pi(\beta) \cup \pi(\alpha) \subset \pi(\nu)$ and that $B$ contains only $e^\circ$, and edges in $G[\beta]$ and in $G[\alpha] \cap \T$. It immediately follows that $B \cap \pi(\beta) = B \cap \pi(\nu)$. 
    \item 
    First, we show how to compute the eastern border.
    Denote by $v$ the boundary vertex of $\alpha$. The vertex $v$ must be the eastern \interval of $B$ in $\pi(\nu)$, together with two edges. 
    Let $e_1$ be the edge in $B \cap G[\alpha]$ that is incident to $v$ and $e_2$ be the spine edge of $\beta$ that is incident to $v$. 
    By Lemma~\ref{lemma:pathdetection}, we test if $e_2$ is contained in a biconnected component $B^* \in BC(\beta, \des_v(\beta))$ in $O(k \log n)$ time. 
    If such a $B$ exists, then the eastern \interval is $v$ with two edges: $e_1$ and an edge of the eastern \interval of $B^*$. 
    If no such $B$ exists then the eastern \interval is $v$ with $e_1$ and $e_2$. 
    
    Next, we compute the western border. Denote by $w$ the vertex in $\beta$ that is incident to $e^\circ$. Denote by $m$ the meet between $w$ and the boundary vertices of $\beta$ and by $e^*$ the spine edge incident to $m$ that is closest to $\alpha$. We obtain these objects in $O(\log n)$ time (Theorem~\ref{thm:meet}). The edge $e^*$ must be in $B$. 
    By Lemma~\ref{lemma:pathdetection}, we test if $e^*$ is contained in a $B^* \in BC(\beta, \des_v(\beta))$ in $O(k \log n)$ time. 
    Since $B^*$ is a biconnected component in $G[\beta]$, it must be that $B^* \subset B$. If such a $B^*$ exists, it must be unique and the western \interval of $B$ is the western \interval of $B^*$. 
    If no such $B^*$ exists then the western \interval of $B^*$ is $m$ with $e^*$ and the last edge on the path from $w$ to $m$ (in $\T$).
\end{enumerate}

\textbf{Case 3.3: $\nu, \alpha$ and $\beta$ are path clusters where $\nu$ and $\alpha$ share a boundary vertex $u$ and $\nu$ and $\beta$ share a boundary vertex $v$. }
\begin{enumerate}[(a), noitemsep, nolistsep]
    \item Denote by $w$ the endpoint of $e^\circ$ in $\beta$.
    Denote by $m$ the meet between $w$ and the boundary vertices of $\beta$ and by $e^*$ the edge incident to $m$ that is closest to $\alpha$.
    We obtain $m$ and $e^*$ in $O(\log n)$ time (Theorem~\ref{thm:meet}).
    The edge $e^*$ must be in $B$. 
    
    Denote by $x = \pi(\alpha) \cap \pi(\beta)$. Using  Lemma~\ref{lemma:pathdetection}, we can detect if there exists a biconnected component $B^* \in BC(\beta, \des_x(\beta))$ in $O(k \log n)$ time. 
    If such $B^*$ exists it must be unique. Moreover, since $B^*$ is a biconnected component of $G[\beta]$, $B^*$ must be a subset of $B$. 
    If $B^*$ exists, then the path $B \cap \pi(\beta)$ goes from the boundary vertex $v$ of $\beta$ to the endpoint of $B^* \cap \pi(\beta)$.
    If no such $B^*$ exists then the path $B \cap \pi(\beta)$ goes from $v$ to $m$. 
    \item The path $B \cap \pi(\nu)$ is equal to the path $B \cap \pi(\beta)$ concatenated with $B \cap \pi(\alpha)$. The first path was computed in (a). The second path can be computed as follows: denote by $z$ the endpoint of $e^\circ$ in $\alpha$ and by $m'$ the meet between $z$ and the boundary vertices of $\alpha$. We compute $m'$ in $O(\log n)$ time (Theorem:meet). Since $B$ can only contain edges in $G[\alpha]$ that are also in $\T$, the path $B \cap \pi(\alpha)$ goes from $m'$ to $x$.
    \item Let $\alpha$ be east of $\beta$. The other case is symmetrical. 
    The eastern \interval is $m'$ with the two edges of $B$ that are incident to $m$.
    We compute the western \interval analogue to Case 3.2 (c).
    \end{enumerate}

\textbf{Case 3.4: $\nu$ is a path cluster and $\alpha$ and $\beta$ are point clusters (Four-way merge).}
\begin{enumerate}[(a), noitemsep, nolistsep]
    \item The path $B \cap \pi(\beta)$ must be equal to the boundary vertex of $\beta$.
    \item Similarly, the path $B \cap \pi(\nu)$ must be equal to the boundary vertex of $\beta$. 
    \item Per definition, the \interval of $B$ in $\nu$ is empty.
\end{enumerate} \end{proof}

\begin{figure}[h]
  \centering
  \includegraphics{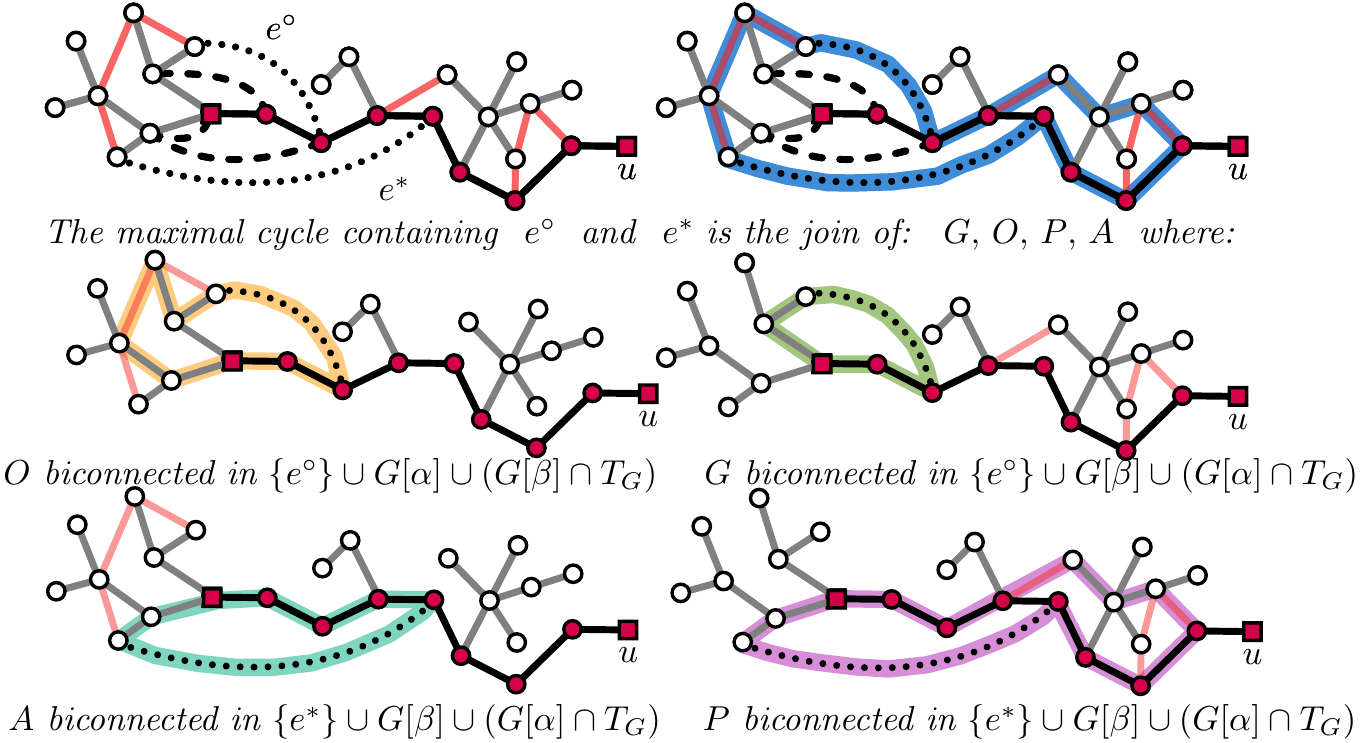}
  \caption{
A cluster $\nu$ with as children a point cluster $\alpha$ and a path cluster $\beta$. There may be many edges in $G[\alpha \cup \beta]$.
These edges are all contained in some maximal cycle which we show in blue. 
For the edge $e^\circ$, we show the biconnected component $G$ in $\{ e^\circ \} \cup G[\beta] \cup \left ( \T \cap G[\alpha] \right)$  and $O$ in $\{ e^\circ \} \cup G[\alpha] \cup \left ( \T \cap G[\beta] \right)$. 
Similarly, for the edge $e^*$ we show the biconnected components $P$ and $A$. On an intuitive level, the maximal blue cycle is their `join'.   
}
  \label{fig:projectedcomponent}
\end{figure}

\subsection{Proving our main theorem}
\label{sub:finalargument}
Finally, given our prerequisite theorems, we show our main result. 
We show that we maintain our invariants in $O(\log ^2 n)$ time per update. The theorem then almost immediately follows. Indeed, to answer biconnectivity queries between vertices $u$ and $v$ we expose them in $O(\log^2 n)$ time.
Denote by $\mu$ the new root and by $\nu$ the path cluster child of $\mu$ and by $\alpha$ and $\beta$ the (at most two) point clusters.
Per definition, $u$ and $v$ cannot be biconnected in $G[\nu]$ (as  the only edges incident to $u$ and $v$ in $G[\nu]$ are in $\pi(\nu)$). 
Similarly, $u$ and $v$ cannot be biconnected in $G[\alpha]$ and $G[\beta]$ as these two graphs do not contain $v$ and $u$ respectively.
Thus, $u$ and $v$ can only be biconnected via a relevant biconnected component $B \in G[\alpha \cup \nu \cup \beta]$. We show that there can be only constantly many interesting such biconnected components and that we can identify then in $O(\log^2 n)$ time, using the same technique we use to maintain our invariants.
This argument upper bounds the aforementioned integer $k$ by a constant.

\theoremmain*

\begin{proof}
We show that we maintain our invariants in $O(\log^2 n)$ time per update. The theorem then almost immediately follows: indeed, to answer biconnectivity queries between vertices $u$ and $v$ we expose them in $O(\log^2 n)$ time.
Denote by $\mu$ the new root and by $\nu$ the path cluster child of $\mu$ and by $\alpha$ and $\beta$ the (at most two) point clusters.
Per definition, $u$ and $v$ cannot be biconnected in $G[\nu]$ (as in $G[\nu]$ the only edges incident to $u$ and $v$ are in $\pi(\nu)$). 
Similarly, $u$ and $v$ cannot be biconnected in $G[\alpha]$ and $G[\beta]$ as these two graphs do not contain $v$ and $u$ respectively.
Thus, $u$ and $v$ can only be biconnected via a relevant biconnected component $B \in G[\alpha \cup \nu \cup \beta]$. We show that there can be only constantly many interesting such biconnected components and that we can identify then in $O(\log^2 n)$ time, using the same technique we use to maintain our invariants. 

Holm and Rotenberg show in~\cite{holm2017dynamic} that any of the update operations can be realized by $O(\log n)$ splits and merges in the top tree. 
Thus, all we have to show is that we can maintain our invariants during splits and merges in the top tree. 

We note that maintaining Invariant~\ref{inv:upwardpointers} can be done in $O(1)$ additional time through standard pointer management during the splits and merges. 
Our argument therefore focuses on maintaining Invariants~\ref{inv:relevant_bicomp} and \ref{inv:component_storage}. Suppose that for each cluster $\nu$ and each vertex $u$, the set $BC^*_u(\nu)$ has at most $k$ elements. 
Then when splitting a cluster $\nu$ with boundary vertex $u$.  all objects representing components in $BC^*_u(\nu)$ and their \interval information can be removed in $O(k)$ time. 

What remains to show, is that $k$ is a constant and that for each merge that creates a cluster $\nu$, we can identify the existence of each $B \in BC^*_u(\nu)$ and compute the \interval of $B$ in $\nu$ in $O(\log n)$ time. 

\textbf{Defining gap closers.}
Any biconnected component $B \in BC^*_u(\nu)$ must contain at least one edge $e \in G[\alpha \cup \beta] \setminus ( G[\alpha] \cup G[\beta])$ for two children $\alpha$ and $\beta$ of $\nu$.
We call such edges \emph{gap closers}. 
We show that for each merge type, there are at most a constant number of gap closers that are interesting (i.e. can be part of a unique biconnected component $B \in BC^*_u(\nu)$ for a boundary vertex $u$ of $\nu$). 
The proof is a case distinction between three cases, that mirror the cases of the proof of Theorem~\ref{thm:merging}. 
The cases are illustrated by Figures~\ref{fig:pointclustermerge}, \ref{fig:pathpointmerge}, \ref{fig:pathpathmerge}, \ref{fig:pathingmerges} and \ref{fig:specialmerge}.

\begin{figure}[h]
  \centering
  \includegraphics{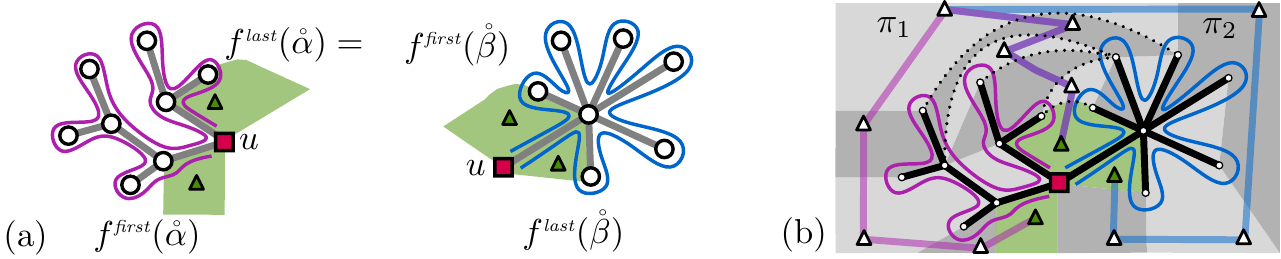}
  \caption{ Case 1.
    (a) Point clusters $\{ \alpha, \beta \}$. We show the faces $(\first{\tour{\alpha}}, \last{\tour{\alpha}}, \first{\tour{\beta}}, \last{\tour{\beta}} )$ in the graph $G$ in green. 
    (b) A schematic representation of cotree vertices (triangles). Any edge from $G[\alpha]$ to $G[\beta]$ must be on $\pi_1 \cap \pi_2$. Only one edge in $\pi_1 \cap \pi_2$ can be incident to $\des_u(\nu)$.
  }
  \label{fig:pointclustermerge}
\end{figure}

\textbf{Case 1: A point merge.}
If $\alpha$ and $\beta$ are point clusters, then any gap closer $e$ must intersect both $\tour{\alpha}$ and $\tour{\beta}$.
Denote by $\pi_1$ the path in $\D$ that coincides with $\tour{\alpha}$ and by $\pi_2$ the path in $\D$ that coincides with $\tour{\beta}$.
It follows, that all gap closers $e$ must lie on the intersection between these two paths: $\pi_1 \cap  \pi_2$ (refer to Figure~\ref{fig:pointclustermerge}. This concept is similar to the \emph{edge bundles} by Laporte et al. in \cite{ItalianoPR93}). 
We can identify $\pi_1 \cap \pi_2$ and a pointer to its first edge $e^\circ$ in $\D$ in $O(\log n)$ time, using the meets between the faces incident to the start and end of $\tour{\alpha}$ and $\tour{\beta}$ (Theorem~\ref{thm:meet}).

For any edge $e$ on $\pi_1 \cap \pi_2$, it is part if a biconnected component $B \in BC^*_u(\nu)$ only if $e^\circ \in B$ (because $e^\circ$ intersect the largest interval of the concatenated tourpaths $\tour{\alpha}$ and $\tour{\beta}$). 
Given $e^\circ$, we immediately apply Theorem~\ref{thm:merging} to identify if there exists a relevant and biconnected component $B$ in $G[\alpha \cup \beta]$ that contains $e^\circ$.
Moreover, using property $(b)$, we can immediately check in $O(1)$ time if $B$ is relevant.
Finally, we test if either $e^\circ$ is edge-incident to $\des_u(\nu)$ in $O(\log n)$ time (Theorem~\ref{thm:edgeincident}). The component $B$ is alive with respect to $\des_u(\nu)$ if and only if it is.

\begin{figure}[h]
  \centering
  \includegraphics{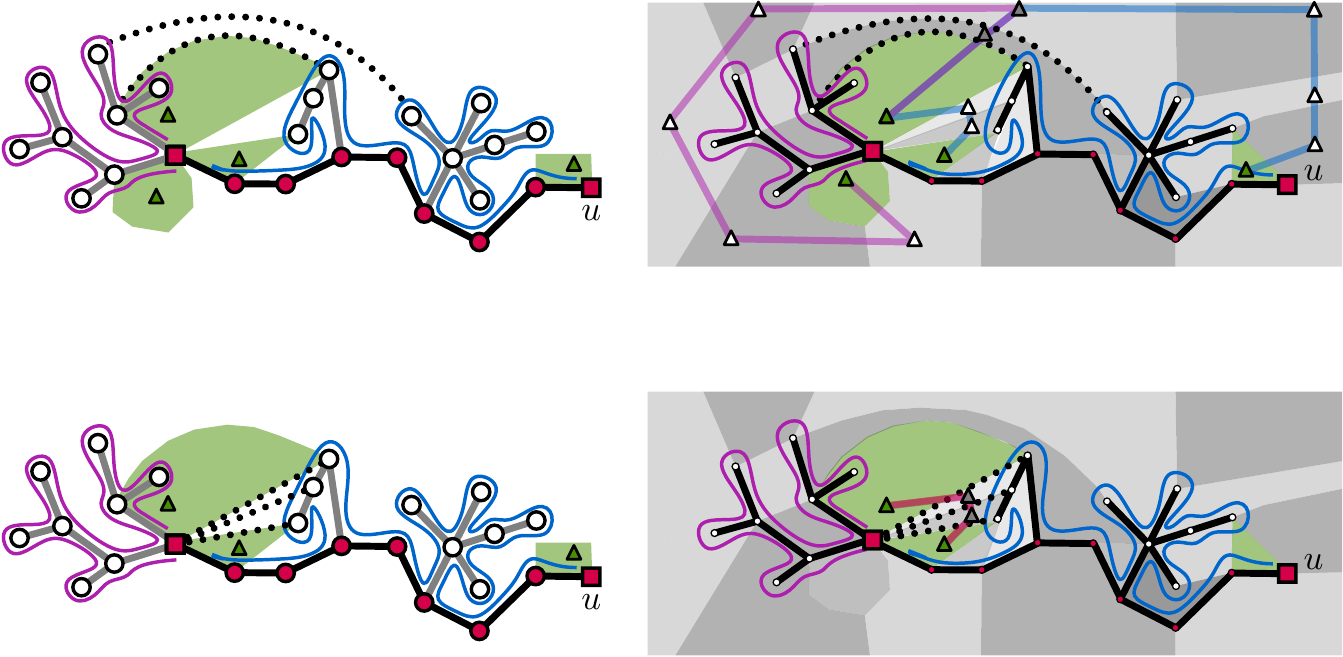}
  \caption{
  Case 2. A path and point cluster $\alpha$ and $\beta$ that form a point cluster (end merge)  \newline
  (top): we illustrate case i for the new edges in $G[\alpha \cup \beta]$. All non-tree edges that intersect $\tour{\alpha}$ (and that not already lie in $G[\alpha])$ must lie on the purple path in the cotree. \newline
  (bottom): we illustrate case ii for the new edges in $G[\alpha \cup \beta]$: edges that have both endpoints in $\beta$ where one endpoint is $\pi(\alpha) \cap \pi(\beta)$. All these edges must lie on the cotree from the `last' face incident to $G[\alpha]$ to the 'first' face incident to $G[\beta]$.
  }
  \label{fig:pathpointmerge}
\end{figure}

\textbf{Case 2: An end merge.}
We refer to Figure~\ref{fig:pathpointmerge}.
Let $\alpha$ be a point cluster and $\beta$ be a path cluster and denote $x = \pi(\alpha) \cap \pi(\beta)$. 
Denote by $u$ the other boundary vertex of $\beta$.
Recall that according to our definition of a graph induced by a path cluster, the graph $G[\beta]$ consists of all edges in $G$ that have both endpoints in $\beta \setminus \{ x, u \}$. 
The graph $G[\nu] = G[\alpha \cup \beta]$ consists of all edges that have both endpoints in $\alpha \cup \beta$. 
We show that this implies that any gap closer $e$ (any edge $e \in G[\alpha \cup \beta] \setminus ( G[\alpha]\cup G[\beta])$) is one of three types. Moreover, we show how to classify these types by  which paths in the Euler tour are intersected by $e$. The endpoints of $e$ lie on either:

\begin{enumerate}[i, noitemsep, nolistsep]
    \item  $\alpha \setminus x$ and $\beta \setminus x$
    (thus $e$ intersects $\tour{\alpha}$ and one of  $\{ \tour{\beta}^\uparrow, \tour{\beta}^\downarrow \}$),
    \item  $\{ \beta \setminus x \}$ and $\{ x \} $  
    (thus $e$ intersects one of $\{ \tour{\beta}^\uparrow, \tour{\beta}^\downarrow \}$ and either: the tourpath $\tour{z}$ connecting $\tour{\alpha}$ to $\tour{\beta}^\uparrow$, or, the tourpath $\tour{z'}$ connecting $\tour{\alpha}$ to $\tour{\beta}^\downarrow$),
      \item $\{ x \}$ and $ \{ z \}$ (there can only be one such edge which we can obtain in $O(\log n)$ time.
\end{enumerate}
  
  \noindent
It follows that any gap closer in $G[\alpha \cup \beta]$ must lie on the intersection $\pi_1 \cap \pi_2$ where $\pi_1$ and $\pi_2$ are paths in $\D$ that coincide with either: $\tour{\alpha}$, $\tour{\beta}^\uparrow$, $\tour{\beta}^\downarrow$, $\tour{z}$ or $\tour{z'}$. 

There are at most a constant number of such pairs $(\pi_1, \pi_2)$. Denote for each pair $(\pi_1, \pi_2)$ by $e^\circ$ the first edge on $\pi_1 \cap \pi_2$. Just as in Case 1, any edge $e \in \pi_1 \cap \pi_2$ is in a biconnected component $B \in BC^*_u$ only if $e^\circ \in B$. 
It follows that the set $BC^*_u(\nu)$ contains at most a constant number of elements and that for each element $B \in BC^*_u(\nu)$ we have access to such an edge $e^\circ \in B$ on a path  $\pi_1 \cap \pi_2$ in $G[\alpha \cup \beta]$.
We now immediately apply Theorem~\ref{thm:merging} to obtain for each such edge $e^\circ$ its `projected biconnected component' $B^\circ$ in the graph $\{ e^\circ \} \cup G[\beta] \cup \left ( \T \cap G[\alpha] \right)$ (and a projected component in $\{ e^\circ \} \cup G[\beta] \cup \left ( \T \cap G[\alpha] \right)$) in $O(k \log n)$ time.
For any pair of such `interesting' gap closers $(e^\circ_1,  e^\circ_2)$, the two gap closers $A$ and $B$ are biconnected in $G[\alpha \cup \beta]$ if either: (i) their projected biconnected components share an edge on $\pi(\nu)$ or (ii) their projected biconnected components share a border on $\pi(\nu)$. Moreover, the radial interval between their border edges must overlap. Or, (iii) they are biconnected through some projected compent $C$ of some gap closer $e^\circ_3$. 
By Theorem~\ref{thm:merging} we have access to all the information to, for each pair of projected components, detect case (i) and case (ii). Thus, we greedily pairwise combine the projected components to detect all biconnected components $B^*$ in $G[\nu]$ (and their borders in $G[\nu]$) and whether they are relevant in $\nu$. We refer to such an operation as a \emph{join}. 

Finally, all that remains is to show that we can test for each relevant maximal biconnected component $B^*$ in $G[\nu]$ if they are edge-incident to the face $\des_u(\nu)$. There are two ways $B^*$ can be edge-incident to $\des_u(\nu)$. Firstly, $B^*$ could contain a gap closer $e^\circ$ that is edge-incident to $\des_u(\nu)$. We verify this for every gap closer in $O(\log n)$ total time (Theorem~\ref{thm:edgeincident}). 
If no gap closer is edge-incident to $\des_u(\nu)$ then it must be that $\des_u(\nu)$ is already enclosed by some cycle in $G[\beta]$.
In this case, $B^*$ is edge-incident to $\des_u(\nu)$ if and only if it contains a biconnected component $B \in BC(\beta, \proj_u(\beta))$.
Such a biconnected component $B$ must contain the spine edge incident to the border of $B$ and so we identify if such a component exists in $O(k \log n)$ time through Lemma~\ref{lemma:pathdetection}.

\begin{figure}[h]
  \centering
  \includegraphics{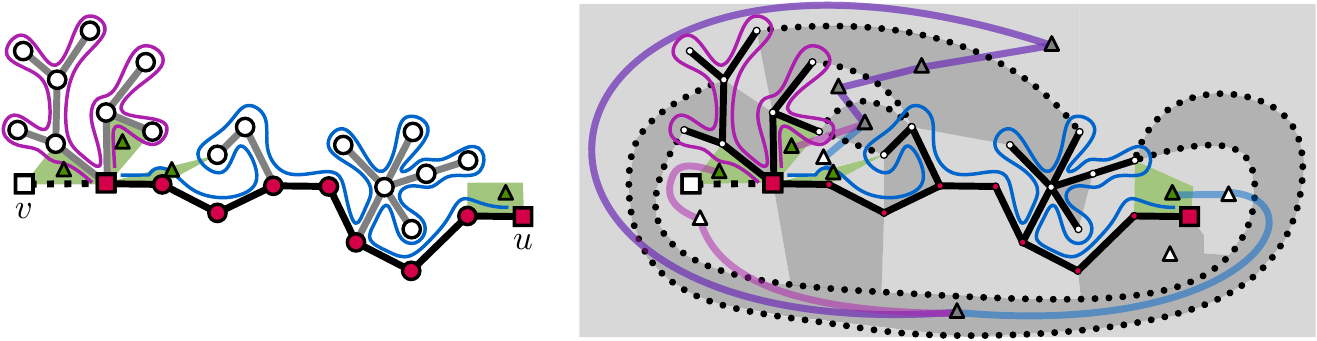}
  \caption{
  Case 3.
   A path and point cluster $\gamma$ and $\beta$ are part of a path cluster. 
  }
  \label{fig:pathpathmerge}
\end{figure}

\textbf{Case 3-5: A four-way merge. }
Let $\nu$ be a path cluster where its children are two path clusters $\alpha$ and $\beta$, and possibly two additional point clusters $\gamma$ and $\zeta$.
Denote by $u$ and $v$ the boundary vertices of $\nu$ and by $m$ its central vertex.
Just as in Case 2, we can categorize the gap closers $e$ of $\nu$ based on which Euler tours they intersect. 
However, for brevity, we describe all cases more high-level.
Any such edge $e$ must have an endpoint:

\begin{enumerate}[i, noitemsep, nolistsep]
\item on a point cluster, and on a path cluster but not on $u, v$ or $m$ (Figure~\ref{fig:pathpathmerge}),
    \item on two path clusters but not on $u, v$ or $m$ (Figure~\ref{fig:pathingmerges}),
    \item on two point clusters but not $m$ (Figure~\ref{fig:specialmerge}), or
    \item on a path cluster and $m$.
\end{enumerate}

\begin{figure}[h]
  \centering
  \includegraphics{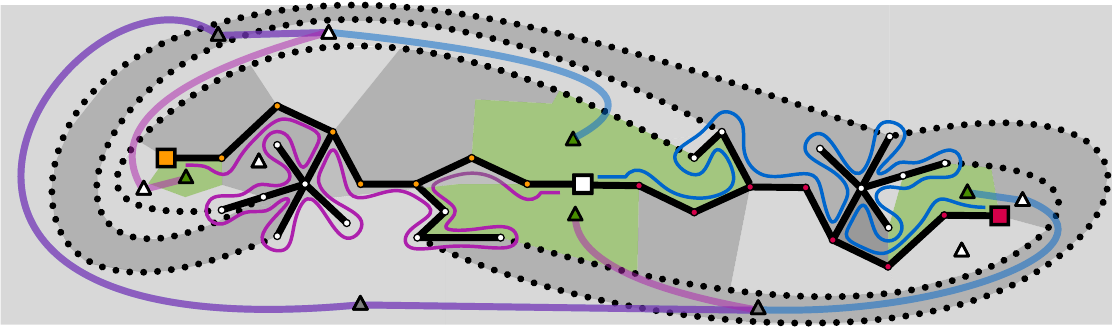}
  \caption{
  Case 4.
    (a) Two path clusters $\alpha$ and $\beta$ that are part of a path cluster.
  } 
  \label{fig:pathingmerges}
\end{figure}

It follows from the above classification, that any gap closer in $G[\nu]$ must lie on the intersection $\pi_1 \cap \pi_2$ where $\pi_1$ and $\pi_2$ are paths in $\D$ that coincide with either: $\tour{\alpha}^\uparrow$, $\tour{\alpha^\downarrow}$, $\tour{\beta}^\uparrow$, $\tour{\beta}^\downarrow$, $\tour{\gamma}$, $\tour{\beta}$ or, alternatively, on any of the four paths in the Euler tree that are incident to $m$ and connect two of these aforementioned tourpaths (these are the tourpaths that detect edges that are incident to $m$ but not in $G[\alpha] \cup G[\beta] \cup G[\gamma] \cup G[\zeta]$).
Just as in Case 2, this implies that we only have to consider at most a constant number of pairs of paths $(\pi_1, \pi_2)$ in the Euler tree and only the first edge of their intersection $\pi_1 \cap \pi_2$. 
To each of these edges $e^\circ$, we apply Theorem~\ref{thm:merging} to obtain constantly many `projected biconnected components'.
This proves that the number $k$ of Theorem~\ref{thm:merging} is a constant and $O(k \log n) = O(\log n)$, 

All biconnected components $B \in BC^*_u(\nu)$ are the join of constantly many projected components. What remains, is to correctly compute the join of these projected components.
To this end, we make one final case distinction.
If there exists no edge from $G[\gamma]$ to $G[\zeta]$, then we greedily compute for every pair of projected  components their join in $O(\log n)$ time by comparing their borders in $G[\nu]$ (the argument that we can compute the join of each pair of projected components in $O(\log n)$ time is identical to the argument in Case 2).

However, if there exists an edge $e^\circ$ from $G[\gamma]$ to $G[\zeta]$, we need to be a bit more careful (Figure~\ref{fig:specialmerge}). Theorem~\ref{thm:merging} gives an empty border of the projected biconnected component $B^\circ$ of $e^\circ$ (this is an artifact from the fact that, in this specific case, we are not able to compute the biconnected component in $G[\gamma \cup \zeta]$ that contains $e^\circ$). This prevents us from immediately computing the join between projected component (as two such components $A$ and $C$ may be biconnected through $B^\circ$). 
We show in Section~\ref{sec:specialcase} Theorem~\ref{thm:separation} to compute if removing $m$ separates $\pi(\alpha)$ from $\pi(\beta)$ in  $G[\nu]$ in $O(\log n)$ time. 

\begin{figure}[h]
  \centering
  \includegraphics{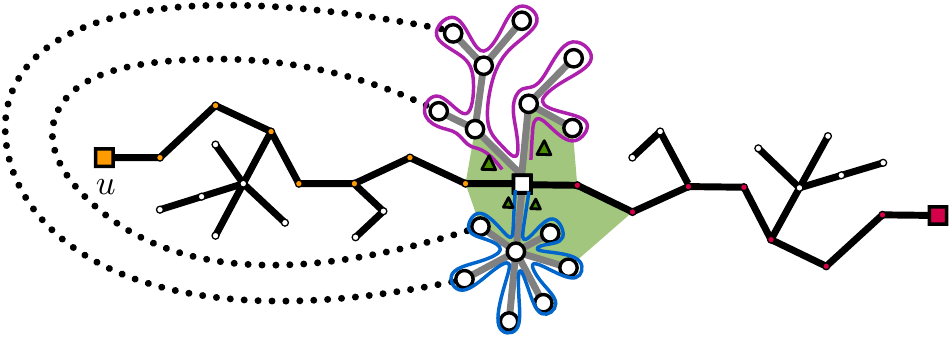}
  \caption{
  Case 5. Two point clusters $\gamma$ and $\zeta$ that share at least one edge.
  } 
  \label{fig:specialmerge}
\end{figure}

Suppose that removing $m$ separates $\pi(\alpha)$ from $\pi(\beta)$ in $G[\nu]$. Then we know that the two edges $e_1$ and $e_2$ on $\pi(\nu)$ incident to $m$, are not biconnected in $G[\nu]$. 
Thus, any two projected components $A$ and $C$ where $A$ contains $e_1$ and $C$ contains $e_2$ cannot be biconnected. 
Similarly, if removing $m$ not separate $\pi(\alpha)$ from $\pi(\beta)$ then $e_1$ and $e_2$ are biconnected in $G[\nu]$. 
Thus, any two projected components $A$ and $C$ where $A$ contains $e_1$ and $C$ contains $e_2$ must be biconnected. 
This allows us to compute the join between two pairs of projected components (and whether they form a biconnected component that is \emph{relevant} to $\nu$) in $O(\log n)$ time, analogue to Case 2. 

Finally, that remains is to show that we can test for each relevant maximal biconnected component $B^*$ in $G[\nu]$ if they are edge-incident to the face $\des_u(\nu)$. There are two ways $B^*$ can be edge-incident to $\des_u(\nu)$. Firstly, $B^*$ could contain a gap closer $e^\circ$ that is edge-incident to $\des_u(\nu)$. We verify this for every gap closer in $O(\log n)$ total time (Theorem~\ref{thm:edgeincident}). 
If no gap closer is edge-incident to $\des_u(\nu)$ then it must be that $\des_u(\nu)$ is already enclosed by some cycle in $G[\beta]$.
In this case, $B^*$ is edge-incident to $\des_u(\nu)$ if and only if it contains a biconnected component $B \in BC(\beta, \proj_u(\beta))$.
Such a biconnected component $B$ must contain the spine edge incident to the border of $B$ and so we identify if such a component exists in $O(k \log n)$ time through Lemma~\ref{lemma:pathdetection}.

\paragraph*{Upper bounding $k$.}
It follows from the above case distinction that for every vertex $u$, for every node $\nu \in \TT$, the number $k$ if biconnected components in the set $B_u^*(\nu)$ is at most a constant (specifically, the constant $k$ is upper bound by the number of pairs that can be selected in a set of ten paths in the Euler tree: $\tour{\alpha}^\uparrow$, $\tour{\alpha^\downarrow}$, $\tour{\beta}^\uparrow$, $\tour{\beta}^\downarrow$, $\tour{\gamma}$, $\tour{\beta}$ or, alternatively, on any of the four paths in the Euler tree that are incident to $m$ and connect two of these aforementioned tourpaths).
Moreover, for every merge we show how to identify these biconnected components and their borders in $O(\log n)$ time.
By the previous result of Holm and Rotenberg~\cite{holm2017dynamic}, we maintain Invariants~\ref{inv:relevant_bicomp}, ~\ref{inv:upwardpointers} and \ref{inv:component_storage} in $O(\log^2 n)$ worst case time per update operation. Indeed, Invariants~\ref{inv:relevant_bicomp}, \ref{inv:upwardpointers} and \ref{inv:component_storage} in $O(\log n)$ additional time per merge in the top tree. Which, by Holm and Rotenberg~\cite{holm2017dynamic} implies a worst case update time of $O(\log^2 n)$ per update operation in $G$.

\paragraph{Answering biconnectivity queries via a root merge.}
Finally we show how to answer biconnectivity queries for two vertices $u$ and $v$.
We expose $u$ and $v$ in $O(\log^2 n)$ time. 
The final root merge, creates the root $\mu$ through merging a path cluster $\nu$ and at most two point clusters $\alpha$ and $\beta$.
Any edge between the two point clusters must create a cycle that contains $u$ and $v$. Such edges must intersect both $\tour{\alpha}$ and $\tour{beta}$. Thus, we can easily detect if any such edge exists and if so, conclude that $u$ and $v$ are biconnected. 

If no such edge exists, we consider the merge between $\nu$ and $\alpha$ (resp. $\nu$ and $\beta$) is identical to case 2. From the analysis of case 2, it follows that we obtain at most a constant number of edges $e^\circ$ that may be part of a maximal biconnected component in $G[\mu]$. 
For each of these, we can obtain their projected component into $\nu$.
The vertices $u$ and $v$ are biconnected if and only if the join of their projections in to $\nu$ covers $\pi(\nu)$. 
This concludes the theorem.
\end{proof}

\section{Testing if a central vertex separates the spine.
}
\label{sec:specialcase}

In this section, we peek into the black box of the Holm and Rotenberg~\cite{holm2017dynamic} data structure to show that we can test for any central vertex $m$, if its removal separates the spine in $G$, in logarithmic time.

\begin{restatable}{theorem}{separation}
\label{thm:separation}
Let $\nu$ be a path cluster with boundary vertices $u$ and $v$. Let $m$ be the central vertex of the merge of its children. 
We can decide if removing $m$ separates $u$ and $v$ in the graph $G[\nu]$ in $O(\log n)$ time. 
\end{restatable}

\begin{proof}
    The structure from Holm and Rotenberg~\cite{holm2017dynamic} is based on a slim-path top tree $\TT$ over $\T$, where each merge and split updates a secondary top tree over $\D$. We augment $\TT$ to also maintain our invariants. An important part of the structure in~\cite{holm2017dynamic} is an operation that in $O(\log n)$ time \emph{covers} all corners along a segment of the extended Euler tour, and an operation that given a path $\pi$ in $\D$, in $O(\log n)$ time finds an internal face on $\pi$ that is incident to an uncovered corner on both sides of the path if such a face exists.  We can easily extend this structure to support a search for an internal face on $\pi$ that is incident to an uncovered corner on at least one side (instead of on both sides).
    
    Observe that $m$ separates $u$ and $v$ if and only if there exists a face $f_m$ in $G[\nu]$ that is incident to corners of $m$ on both the north and south side of $\pi(\nu)$.
    That face may be either $\des_u(\nu)$, $\des_v(\nu)$, or some (other) face that exists in both $G$ and $G[\nu]$.
    
    Now $f_m=\des_u(\nu)$ if and only if in $\D$ the path $\first{\tour{\nu}^\uparrow}\cdots\first{\tour{\nu}^\downarrow}$  both contains a face incident to a corner of $m$ that is north of the spine, and a corner of $m$ that is south of the spine. 
    We can check each of these questions for the path in worst case $O(\log n)$ time by asking first with everything but $\tour{\nu}^\uparrow$ temporarily covered and then with everything but $\tour{\nu}^\downarrow$ temporarily covered.
    The case $f_m=\des_v(\nu)$ is symmetric.
    
    Finally, a face that is in both $G$ and $G[\nu]$ can only be $f_m$ if it is an internal face on the common path in $\D$ between $\first{\tour{\nu}^\uparrow}\cdots\last{\tour{\nu}^\uparrow}$ and $\first{\tour{\nu}^\downarrow}\cdots\last{\tour{\nu}^\downarrow}$.  We can find the ends $f_1,f_2$ of this path in worst case $O(\log n)$ time using the meet operation, temporarily cover everything except $\tour{\nu}^\uparrow$ and $\tour{\nu}^\downarrow$, and then search $f_1\cdots f_2$ for an internal face that is incident to $m$ on both sides of the path in worst case $O(\log n)$ time.
\end{proof}

\bibliography{refs}
\appendix

\end{document}